\documentclass[12pt]{iopart}
\usepackage{bm}
\usepackage{array}
\usepackage{graphicx}
\usepackage{color}
\expandafter\let\csname equation*\endcsname\relax
\expandafter\let\csname endequation*\endcsname\relax
\usepackage{nicematrix}

\DeclareMathAlphabet{\mathdj}{U}{msb}{m}{n}  
\newcommand{\Reals}{\ensuremath{\mathdj {R}}} 

\newbox\squ  

\setbox\squ=\hbox{\vrule width.4pt 
   \vbox{\hrule height.4pt width.6em\kern1.4ex 
         \hrule height.4pt}%
   \vrule width.4pt depth0pt }
\def\sqbox{\copy\squ\hskip -.4pt}

\newenvironment{proof}{\topsep=\smallskipamount \partopsep=0pt  %
 \begin{trivlist} \itemindent=\parindent                        %
 \item[\hskip \labelsep\emph{Proof:}]}{\qed\end{trivlist}}      %
\let\qed=\relax                                                 %
\def\qed                                                        %
 {{\unskip\nobreak\hfil\penalty50                               %
   \quad\hbox{}\nobreak\hfil $\sqbox$                             %
   \parfillskip=0pt \finalhyphendemerits=0 \par}}               %

\usepackage{enumitem}

\usepackage{float}
\restylefloat{table}

\newtheorem{theorem}{Theorem}[section]
\newtheorem{corollary}{Corollary}[theorem]
\newtheorem{lemma}[theorem]{Lemma}
\newtheorem{prop}{Proposition}

\newcommand{\INF}{\mathrm{INF}}
\newcommand{\Collins}{\mathrm{C}}
\newcommand{\LK}{\mathrm{LK}}
\newcommand{\RT}{\mathrm{RT}}

\begin{document}

\title[$G_2$ Cosmologies II:  One Hypersurface-Orthogonal Killing Vector Field]{Dynamical Equilibrium States of a Class of Irrotational Non-Orthogonally Transitive $G_{2}$ Cosmologies II: Models With One Hypersurface-Orthogonal Killing Vector Field}

\author{ Sepehr Rashidi$^1$, C.G. Hewitt$^{2,1}$, Benoit Charbonneau$^{3,1}$}

\address{$^1$ University of Waterloo, Department of Physics And Astronomy, Waterloo, Ontario, Canada}
\address{$^2$ University of Waterloo, C.E.M.C., Faculty of Mathematics, Waterloo, Ontario, Canada}
\address{$^3$ University of Waterloo, Department of Pure Mathematics, Waterloo, Ontario, Canada}
\ead{srashidi@edu.uwaterloo.ca$^{1}$}
\ead{cghewitt@uwaterloo.ca$^{2}$}
\ead{benoit@alum.mit.edu$^{3}$}
\vspace{10pt}
\begin{indented}
\item[]March 2021
\end{indented}

\begin{abstract}
We consider a class of inhomogeneous self-similar cosmological models in which the perfect fluid flow is tangential to the orbits of a three-parameter similarity group. We restrict the similarity group to possess both an Abelian $G_{2}$, and a single hypersurface orthogonal Killing vector field, and we restrict the fluid flow to be orthogonal to the orbits of the Abelian $G_{2}$. The temporal evolution of the models is forced to be power law, due to the similarity group, and the Einstein field equations reduce to a three-dimensional autonomous system of ordinary differential equations which is qualitatively analysed in order to determine the spatial structure of the models. The existence of two classes of well-behaved models is demonstrated. The first of these is asymptotically spatially homogeneous and matter dominated, and the second is vacuum dominated and either asymptotically spatially homogeneous or acceleration dominated, at large spatial distances. 
\end{abstract}

%
%
%
%

\section{Introduction:}

In the companion paper to this, \cite{hewitt2020}, we have defined the N-OT $G_{2}$ cosmologies to be those exact solutions of the Einstein Field Equations (EFEs) with an expanding perfect fluid source and linear equation of state, whose symmetry group is an Abelian two-dimensional isometry group which does not necessarily act orthogonally transitively. It is assumed that the reader is familiar with \cite{ellis1Maccullum969class}, or similar works \cite{wainwright2005dynamicalSystems}, where the variables and formalism that we employ and exploit have been established. The goal of this paper is to perform a complete analysis of the dynamical equilibrium states of the subclass of the N-OT $G_{2}$ cosmologies in which the perfect fluid source is irrotational and in which one of the Killing vector fields (KVFs) is hypersurface orthogonal (HO). That is, we examine a class of self-similar solutions of the EFEs in which: there is a three parameter homothety group \cite{eardley1974commun} $H_{3}$ acting transitively on timelike hypersurfaces; the $H_{3}$ has a two dimensional Abelian isometry subgroup $G_{2}$ and one of the KVFs is HO; there is an expanding irrotational perfect fluid source, with linear equation of state flowing orthogonal to the orbits of the $G_{2}$ and tangential to the orbits of the $H_{3}$. 
We begin, in Section 2, with an overview of the previous works that have examined N-OT
$G_{2}$ models admitting one HO KVF. These works were mainly focused on obtaining separable line-elements. We provide a complete classification of the N-OT $G_{2}$ cosmologies with separable line-element admitting one HO KVF, and compare it to the classification of the closely-related diagonal $G_{2}$ cosmologies with separable line-element, in which there are two 
HO KVFs. The ordinary differential equation for the models under consideration here is simply obtained in Section 3 by imposing the appropriate restrictions on the systems developed in \cite{hewitt2020}.
New variables are introduced in order to compactify the phase space, and thus make the qualitative analysis simpler. The invariant sets, equilibrium points and their local stability are also provided.\\
In Section 4 we present a complete analysis of the vacuum models, which arise on the boundary of our phase space. This completes the analysis of all the self-similar models which arose in \cite{van1988class}.
We show that the generic models typically follow one of the following two types of behaviour: they are either well-behaved and asymptotic to a plane wave model at large spatial distance, or they are badly behaved at large spatial distance with some of the components of the Electric and Magnetic parts of the Weyl tensor diverging. The exception is when the (equation of state) parameter is $\frac{4}{3}$, in which almost all of the vacuum models are asymptotically flat at large spatial distance.
In Section 5, we consider perfect fluid models with a linear equation of state relating pressure $p$ to energy density $\mu$ through $p=(\gamma-1)\mu$. We present a complete analysis for a variety of values of the equation of state parameter $\gamma$, namely $1<\gamma<\frac{3}{2}$, $\gamma\neq \frac{10}{9}$. We reveal an interesting monotone function which is crucial in limiting the possible asymptotic behaviours of the models, at large spatial distance. Models which are well-behaved exist for the two ranges of $\gamma$ of, $(1,\frac{10}{9})$ and $(\frac{6}{5},\frac{3}{2})$, and fall into three distinct classes, depending on the value of the equation of state parameter. These classes are, asymptotically spatial homogeneous and matter dominated, asymptotically spatial homogeneous and vacuum dominated ($\frac{10}{7}$ only), and, acceleration dominated and vacuum dominated. As in the analogous exceptional orthogonal Bianchi VI($h=-\frac{1}{9}$) spatially homogeneous class, the value of the equation of state parameter being equal to $\gamma=\frac{10}{9}$ causes a dramatic bifurcation of the entire phase space, and this case merits its own section, namely Section 6. A Stokes--Bendixson--Dulac function is presented, whose existence eliminates the possibility of limit cycles arising in the phase space. For this value of $\gamma$, phase space is foliated by a one-parameter family of invariant two spaces, and we show that there are open sets of models which are asymptotic to the Wainwright Bianchi VI($h=-\frac{1}{9}$) model at large spatial distances. The appendix contains the details about the asymptotic analysis that is summarized in Sections 4 and 5.

\section{Classification Of $G_{2}$ Cosmologies With Separable Line-Element}

The main goal of this paper is to examine the asymptotic states of the N-OT $G_{2}$ cosmologies in which one of the KVFs is HO. We may complete this task without dwelling on the exact form of any associated line-element, and that has always been one of the advantages of the technique of using dimensionless variables associated with a group invariant orthonormal frame. However, to make connections to previous works, we perform an analysis of the line-element in this section. We are aware of three papers that have looked at the exact solutions of the EFEs which admit an Abelian $G_{2}$ which has one HO KVF. Each of the papers commenced with the line-element
\begin{eqnarray}
	ds^{2} = -e^{2k} dt^{2} + e^{2h} dx^{2} + r[f \ dy^{2}+f^{-1}(dz+ \ wdx)^{2}], \label{metricAnztaz}
\end{eqnarray} 
and assumed that the functions appearing in the line-element are separable functions of the dependent variables $t$ and $x$. In each case, the main goal was to obtain exact solutions.
In \cite{van1988class}, only vacuum models are investigated and various systems of ODEs are derived but not analysed. We complete that investigation in Section 5. In \cite{wils1991inhomogeneous}, perfect fluid models are looked at, with the primary goal being to uncover exact solutions. Many of the exact solutions that are obtained are transitively self-similar, and are in the paper of \cite{hsu1986self}, and most of the rest either have singularities or they are, either vacuum, or stiff. \cite{van1992qualitative} analysed these stiff models in more detail using qualitative analysis on an unbounded phase space.\\

Let us first consider the problem of trying to find exact solutions of the \emph{diagonal $G_{2}$} cosmologies with separable line-element, that is, what happens if we set the function $w$ to zero in (\ref{metricAnztaz}) and impose separability on the line-element. There are three main classes of (non-stiff, non-vacuum) well-behaved solutions and these are
\begin{enumerate}
	\item spatially homogeneous,\\
	\item self-similar, \\
	\item models with no additional symmetry.
\end{enumerate}

The line elements in the first two cases are well-known and appear in \cite{kramer1980exact}. \\
The line elements in class (iii) fall naturally into 5 mutually exclusive classes( see 
\cite{hewitt1994invariant} ). These are divided into the Wainwright Goode class and the Unit Shear class as shown in Table \ref{TableofClassIIIModelsWithNoAdditionalSymmetry}.
\begin{table}[h!]
	\caption{Models With No Additional Symmetry}\label{TableofClassIIIModelsWithNoAdditionalSymmetry}
	\lineup
	\begin{center}
	\begin{tabular}{@{}*{3}{l}}
		\br                              
	    Restriction & Class Name   \\ 
		\hline
	    $N_{X}=0$ & Wainwright Goode Class \\
		\hline
	    $\Sigma = 1 \ \gamma = \frac{4}{3}$ & Unit Shear Class  \\
		\hline
	\end{tabular}	
	\end{center}
\end{table}

Models in the Unit Shear class are distinguished by the magnitude of the deceleration parameter $q$ as illustrated in Table \ref{TableofUnitShearClass}.

\begin{table}[h!]
	\caption{Unit Shear Class.}\label{TableofUnitShearClass}
	\lineup
	\begin{center}
	\begin{tabular}{@{}*{5}{l}}
		\br                              
		Class & $q$ & Name & Features \\ 
		\hline
		I & $q < 1$ &  \begin{tabular}{c} Generalized Senovilla \cite{senovilla1990new} \end{tabular} & \begin{tabular}{c} Cylindrically \\ Symmetric \\ No big-bang \end{tabular}
		\\
		\hline
		II & $1<q<2$ &\begin{tabular}{c} Generalized Feinstein--Senovilla \cite{feinstein1989new} \end{tabular} & \begin{tabular}{c} Cylindrically \\ Symmetric \end{tabular}  \\
		\hline
		III & $q=2$ & \begin{tabular}{c} Davidson \cite{davidson1991big} \end{tabular} & \begin{tabular}{c} Cylindrically \\ Symmetric \end{tabular}\\
		\hline 
		IV & $2 < q$ &\begin{tabular}{c} Van den Bergh and Skea \cite{van1992inhomogeneous}, \\ Inhomogeneous generalizations\\
			 of the time-symmetric\\ Kanowski--Sachs model \end{tabular} & \begin{tabular}{c} Spatially\\ Inhomogeneous \\ and periodic\\ recollapse \end{tabular} \\
		\hline
	\end{tabular}	
\end{center}
\end{table}
The models are all asymptotically self-similar at both early and late times except the Davidson model which is asymptotically self-similar at early times only.\\

We now repeat this analysis for the N-OT $G_{2}$ cosmologies with separable line-element admitting one HO KVF. That is, we consider the line-elements (\ref{metricAnztaz}) in the case when $w$ is not zero. Once again there are three classes of model and these are
\begin{enumerate}
	\item spatially homogeneous ( see \cite{davidson1991big} ),
	\item self-similar. These are the very models which we study in this paper,
	\item models with no additional symmetry.
\end{enumerate}
The line-element for case (i) can be written 
\begin{eqnarray}
\fl ds^2 =& -dt^2 + e^{2K(t)}dx^2+ e^{2(K(t) + S(t) + F(t))}e^{8 x} dy^2\\
\fl \phantom{=}&+ e^{2(K(t) + S(t) - F(t))} e^{4 x}(dz + w(t) e^{-2x} dx )^2. \nonumber
\end{eqnarray}
The line element for case (ii) can be written
\begin{eqnarray}
\fl ds^2 = e^{2K(x)} e^{2t}(-dt^2+dx^2) + e^{2(K(t) + S(t) + F(t))} e^{\frac{2}{\gamma}(4-3\gamma)t} dy^{2}  \\
\fl \phantom{ ds^2=} + e^{2(K(t) + S(t) - F(t))} e^{ \frac{4}{\gamma}(\gamma-1)t}(dz + w(x) e^{ \frac{2-\gamma}{\gamma}t } dx )^2. \nonumber
\end{eqnarray}
A careful examination of the third class reveals that their equation of state parameter is restricted by $-7 \gamma^{2} + 16 \gamma -8 = 0$, $\gamma = \frac{8 \pm \sqrt{8} }{7}$. Note that this is the same restriction that appears in \cite[equation 3.18]{wils1991inhomogeneous}. In addition, these models have
\begin{eqnarray}
	\Sigma_{+} = - \frac{3\gamma-2}{4}, \Sigma_{-} = \frac{6-5\gamma}{4}, \mbox{and} \quad \Sigma = 1+ \Sigma_{13}^{2} \geq 1.
\end{eqnarray}
Detailed analysis of this third class shows that it naturally partitions into three subclasses according to whether the deceleration parameter is greater than, equal to, or less than 2. The models in each of these subclasses are badly-behaved in that the acceleration of the fluid diverges on each spacelike hypersurface, $t=constant$. Unfortunately, there are no generalizations of the generalized Senovilla model, generalized Feinstein--Senovilla, Davidson nor the Van den Bergh--Skea models in this class. In addition the simplest well-behaved models which have not been fully examined are the self-similar ones which we discuss here.

\section{ODEs and General Properties}

The EFEs for the models under consideration here are obtained from equation (1) in \cite[equation 1]{hewitt2020} by setting $\Sigma_{23}=0$ and $N_{22}=0$. There are two curvature variables, $A$, $N_{x}$, and a non-constant shear variable $\Sigma_{13}$.

We then have the following 1-parameter, three dimensional, system of ODEs in dimensionless form:
\begin{eqnarray}
\boldsymbol \partial_{1} A = 2 \bigl[ A^{2} - (\frac{3}{4})^{2}(2-\gamma)^{2} \bigl] + \frac{3}{2}(2-\gamma) \Omega + 9 (\Sigma_{13})^{2}, \label{G2H3-ODE-A-N-SIGMA13-1}\\
\boldsymbol \partial_{1}(N_{\times}) = 2 A N_{\times}  - \frac{3}{8}(2-\gamma)(6-5 \gamma) - 3 \Sigma_{13}^{2}, \label{G2H3-ODE-A-N-SIGMA13-2}\\
\boldsymbol \partial_{1} (\Sigma_{13}) =  3\bigl( A + \dot{U} - N_{\times} \bigl) \Sigma_{13}, \label{G2H3-ODE-A-N-SIGMA13-3}
\end{eqnarray}
where the acceleration variable $\dot{U}$ is given by
\begin{eqnarray}
\dot{U} = \frac{1}{3 (2-\gamma)} \bigl[(3\gamma-2)A  - 3(6-5\gamma) N_{\times} \bigl],
\end{eqnarray}
and the dimensionless energy density $\Omega$ is given by
\begin{eqnarray}
\fl 3(\gamma - 1)\Omega &= 9(3-2\gamma)(\gamma - 1)+ \frac{(6-7 \gamma)}{(2-\gamma)} A^{2}  - 9N_{\times}^{2}
+ \frac{6(6-5\gamma)}{(2-\gamma)} A N_{\times} +  9 \Sigma_{13}^{2}.  \label{G2H3-ODE-A-N-SIGMA13-4} 
\end{eqnarray}
The auxiliary equation is
\begin{eqnarray}
\boldsymbol \partial_{1}(\Omega) = 3\Omega  \dot{U} \Bigl(  \frac{2 - \gamma}{1 - \gamma}  \Bigl). 
\end{eqnarray}
The models which have $\Sigma_{13}=0$ are self-similar diagonal $G_{2}$ cosmologies and have been examined in the papers of \cite{hewitt1988qualitativeI} and \cite{hewitt1991qualitativeII}. There are two parameters in these works, the equation of state parameter $\gamma$, which was restricted by $1 < \gamma < 2$, and a shear parameter, labelled $r$, which was restricted to be non-negative due to the symmetry in the models. The two papers \cite{hewitt1988qualitativeI} and \cite{hewitt1991qualitativeII} were distinguished by the range of the variable $r$, with the models in \cite{hewitt1988qualitativeI} satisfying $0 \leq r \leq r^*$, where $(r^*)^{2} = \frac{7\gamma - 6}{3 \gamma - 2}$. The models in \cite{hewitt1991qualitativeII} had $r^* < r$, and it was shown that such models were well-behaved iff $r=r^{**}$, where $(r^{**})^{2}=\frac{3\gamma-2}{2-\gamma}$.\\

It follows from equation (A.68) in \cite{hewitt2020} and \cite{hewitt1997dynamical} that the models considered in this work have $r^{2}=\frac{6-5\gamma}{(2-\gamma)(3\gamma-2)}$. If we restrict $\gamma$ by $1<\gamma<2$ then we obtain generalizations of the models in \cite{hewitt1988qualitativeI} iff $(2\gamma-3) \leq 0 $, and we do not obtain any generalizations of the models in \cite{hewitt1991qualitativeII}. Thus, we restrict the equation of state parameter $\gamma$ by $1 < \gamma < \frac{3}{2}$ in this paper.\\

An examination of the equation for the energy density, (\ref{G2H3-ODE-A-N-SIGMA13-4}), reveals that the vacuum models lie on a hyperboloid of one sheet and then the phase space is unbounded. We introduce the variables $D, S$ and $\chi$ defined by
\begin{eqnarray}
D = 1+ \frac{\Sigma_{13}^{2}}{(\gamma-1)(3-2\gamma)}, \
S = D^{-1/2}, \ \ \frac{d}{d\chi} = \frac{8(\gamma-1)}{3}S  \ \boldsymbol \partial_{1}. \label{S-D}
\end{eqnarray}
Note that $0 \leq S \leq 1$ and $\Sigma_{13} \longrightarrow \infty \iff S \longrightarrow 0$.\\
Also, we scale the variables $A, N_{x}, \Sigma_{13}$ and $\Omega$, to obtain the new variables $Y_{1}, Y_{2}, Y_{3}$, $Y_{4}$:
\begin{eqnarray}
\fl Y_{1} = \frac{4}{3(2-\gamma)} S A, \ Y_{2} = 4 S \bigl( \frac{ (6-5\gamma) A - 3(2-\gamma) N_{x} }{3(2-\gamma)}  \bigl), \ Y_{3} = \frac{4S \Sigma_{13} }{\sqrt{2-\gamma}}, \ Y_{4} = \Omega S^{2}. \label{CoordinateTransTUV-to-Y1Y2Y3Y4}
\end{eqnarray}
Note that if $Y_{4}=0$, then it does not imply that $S = 0$. In terms of these new variables and the independent variable $\chi$, the EFEs can be written as:
\begin{eqnarray}
\fl \frac{d Y_{1}}{d \chi} = \frac{4(\gamma - 1)}{(2 - \gamma)} (1-S^{2}) [ 4 \gamma ( 3 - 2 \gamma) - Y_{1} L ] + 4(10 - 7 \gamma)(\gamma - 1) \ ( S^{2} - Y^{2}_{1} ) - Y^{2}_{2}, \label{ODE-Y1Y2D2-1}\\
\fl \frac{ d Y_{2}}{d \chi} = \frac{4( \gamma - 1) }{(2-\gamma)}(1-S^{2}) \  [ 4(3 - 2 \gamma)(12 \gamma - 4 - 7 \gamma^{2}) - Y_{2} L] + 4(\gamma-1)(2 - \gamma) Y_{1} Y_{2} \nonumber \\
\fl  \phantom{\frac{ d Y_{2}}{d \chi} =}  + (6-5\gamma) \ [16(\gamma-1)(3-2\gamma)(S^{2} - Y^{2}_{1}) - Y^{2}_{2}], \label{ODE-Y1Y2D2-2} \\
\fl \frac{d S}{d \chi} = -\frac{ 4(\gamma - 1)}{(2-\gamma)}[S(1 - S^{2})L].  \label{ODE-Y1Y2D2-3}
\end{eqnarray}
where $L=(-20 + 36 \gamma -15 \gamma^{2})Y_{1} + (4-3 \gamma) Y_{2}$ is a scaled acceleration term. \\
For more details on these calculations please see \cite{rashidi2019subclass}.\\
Note that the differential equation for the shear variable has been replaced by the differential equation for $S$. The scaled energy density, $Y_{4}$, is given by
\begin{eqnarray}
\fl \frac{16(\gamma-1)}{3} Y_{4} =  16(3-2 \gamma)(\gamma - 1) (1-Y_{1}^{2}) - Y^{2}_{2}, \label{ODE-Y1Y2D2-Y4-Equation}
\end{eqnarray}
and the new auxiliary equation is
\begin{eqnarray}
\fl \frac{d Y_{4}}{d \chi} = 2 \Bigl[ (5\gamma - 6) Y_{2} + 4( \gamma - 1)(7 \gamma - 10) Y_{1} - \frac{ 4(\gamma - 1)}{(2- \gamma)}(1-S^{2}) L \Big]Y_{4}.\label{ODE-Y1Y2D2-dY4dchiEquation}
\end{eqnarray}
The system has a reflection symmetry through the origin, given by
\begin{eqnarray}
(\chi,Y_{1},Y_{2},S)   \mapsto (-\chi,-Y_{1},-Y_{2},S).
\end{eqnarray}
The physical region of phase space is the boundary and interior of the truncated elliptical cylinder defined by
\begin{eqnarray}
\fl \frac{16(\gamma-1)}{3} Y_{4} =  16(3-2 \gamma)(\gamma - 1) (1-Y_{1}^{2}) - Y^{2}_{2} \geq 0, \quad 0 \leq S \leq 1, \label{EllipticalCylinder}
\end{eqnarray}
and is illustrated in Fig \ref{FigureCompactified}.
\begin{figure}
	\caption{Compactified Phase Space}
	\centering
	\includegraphics[width=0.4\textwidth]{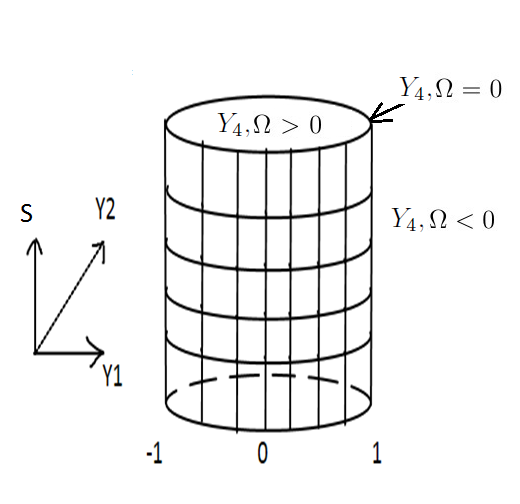} \label{FigureCompactified}
\end{figure}

\subsection{Invariant Sets}

The system possesses three invariant sets:
\begin{enumerate}
	\item $S = 1$. Models in this class admit two HO KVFs, and were considered in \cite{hewitt1988qualitativeI}.\\
	\item $Y_{4}=0$. These are the vacuum models, a differential equation for these models was written down in \cite{van1988class} but was not fully analyzed. We provide a complete analysis of these models here.\\
	\item $S=0$. This invariant set corresponds to unphysical models. Points at infinity in the original variables are now located on this set.
\end{enumerate}
\subsection{Equilibrium Points}
Before we provide the information about the equilibrium points and the eigenvalues of the linearization matrix at each equilibrium point, it is convenient to introduce the following quadratic expressions in $\gamma$. Let
\begin{eqnarray}
	Q_{1}(\gamma) :=  -63 \gamma^{2} + 156 \gamma - 92, \quad \quad \
	Q_{2}(\gamma) := -15 \gamma^{2} + 36\gamma - 20, \nonumber \\
	Q_{3}(\gamma) := -193\,{\gamma}^{2}+412\,\gamma-196, \quad Q_{4}(\gamma) := -171 \gamma^{2} + 408 \gamma - 224. \nonumber
\end{eqnarray}
$Q_{1}(\gamma)$, $Q_{2}(\gamma)$, and $Q_{4}(\gamma)$ are positive for $1 < \gamma < \frac{3}{2}$. $Q_{3}(\gamma)$ is positive for $1 < \gamma < \gamma_{1}$, and negative for $\gamma_{1} < \gamma < \frac{3}{2}$, where $\gamma_{1} = \frac{206}{193} + \frac{48 \sqrt{2}}{193} \approx 1.42$.\\
\begin{table}[H]
	\caption{Equilibrium Points Of The System}\label{TableofEquilibriumPointsGeneral}
	\lineup
	\begin{tabular}{@{}*{5}{l}}
		\br                              
		\small{Equilibrium Point (label)}& \small{$S^{2}$}& \small{$Y^{2}_{1}$}& \small{$Y_{2}$}& \small{$Y_{4}$}   \\ 
		\hline
		\small{Unphysical ($\INF$)}& \small{$0$} & \small{$\frac{4(3 - 2 \gamma )(\gamma - 1)}{(2-\gamma )^{2}}$}& \small{$2(4 - 3\gamma)Y_{1}$}& \small{$0$} \\
		\hline
		\small{Collins$\mathrm{VI}_h$ ($\mathrm{C}$)} & \small{$1$}& \tiny{$\frac{(6 - 5 \gamma)^{2}}{(3 \gamma - 2)(2 - \gamma)}$} & \small{$-4\frac{(10-7\gamma)(\gamma - 1)}{(6-5 \gamma)} Y_{1}$} & \small{$\frac{3(10 - 7 \gamma)(\gamma - 1)}{(3 \gamma - 2)}$}  \\
		\hline
		\small{Plane Waves ($\mathrm{LK}$)} & \small{1} & \small{1} & \small{0} & \small{0} 	\\
		\hline 
		\small{Robinson--Trautman ($\mathrm{RT}$)} & \tiny{$\frac{24(\gamma - 1)(3 - 2 \gamma)}{Q_{1}(\gamma)}$} & \tiny{$\frac{16(\gamma - 1)(4-3\gamma)^{2}(3 - 2\gamma)}{Q_{1}(\gamma)(2 - \gamma)^{2}}$} & \small{$-\frac{Q_{2}(\gamma)}{(4-3\gamma)}Y_{1}$} & \small{$0$} \\ 
		\hline
		\small{Wainwright,$\gamma =\frac{10}{9}$ (W)} & \small{$\frac{7}{3} - 8 Y_{1}^{2}$} & \small{$\frac{1}{6} \leq Y^{2}_{1} \leq \frac{7}{32}$} & \small{$- \frac{20}{9}Y_{1}$} & \small{$\frac{7}{3} - \frac{32}{3}Y^{2}_{1}$} \\
		\hline
	\end{tabular}	
\end{table}
Table \ref{TableofEquilibriumPointsGeneral} provides the information about the equilibrium points of the system. The equilibrium points are transitively self-similar cosmological models; they admit a $H_{4}$, which possess a $H_{3}$ subgroup. The $H_{3}$ has an abelian $G_{2}$ subgroup and the fluid is tangential to the $H_{3}$ orbits.\\

The list of the signs of the eigenvalues of the linearization matrix for each equilibrium point is shown in Table \ref{TableofSignofTheEigenvalues}. \\
\begin{table}[H]
	\caption{Signs of the real parts of the eigenvalues } \label{TableofSignofTheEigenvalues}
	\begin{tabular}{@{}*{5}{l}}
		\br                              
		Equilibrium Point &  $1 < \gamma < \frac{10}{9}$ & $\frac{10}{9} < \gamma < \frac{6}{5}$ & $\frac{6}{5}< \gamma < \frac{10}{7}$ & $\frac{10}{7} < \gamma < \frac{3}{2}$ \\ 
		\hline
		$\mathrm{INF}^{\pm}$& \multicolumn{4}{c}{ $\mp \mp \mp$ } \\ 
		\hline
		$\Collins^{\pm}$ & 
		$\pm \pm \pm$ & 
		$\mp \pm \pm$ & 
		$\mp \mp \mp$ &	$Y_{4}<0$ (Unphysical)\\ 
		\hline
		$\mathrm{LK}^{\pm}$ & \multicolumn{3}{c}{ $\mp \pm \pm$ }  & $ \pm \pm \pm$ \\ 
		\hline 
		$\mathrm{RT}^{\pm}$ &$\pm \pm \mp$ & \multicolumn{3}{c}{ $\mp \pm \mp$ }  \\
		\hline
	\end{tabular}	
\end{table}
The eigenvalues of the linearization matrix at each equilibrium point are provided in Table \ref{TableofEigenvaluesPlusMinus}.
\begin{table}[H]
	\caption{Eigenvalues of the linearization matrix when Sign($Y_{1}$) $= \pm 1$.}\label{TableofEigenvaluesPlusMinus}
	\begin{center}
		\lineup
		\begin{tabular}{@{}*{5}{l}}
			\br                              
			Equilibrium Point &  The eigenvalues of the linearization matrix   \\ 
			\hline
			$\mathrm{INF}^{\pm}$	
			& \rule{0pt}{4ex}  $\lambda_{1}= \mp 24\, ( \gamma-1 ) ^{ \frac{3}{2}}\sqrt {3 - 2\,\gamma}$ \\
			& \rule{0pt}{4ex} $\lambda_{2} = \mp (16)(\gamma-1)^{ \frac{3}{2} }\sqrt{3-2\gamma}$ \\ 
			& \rule{0pt}{4ex} $\lambda_{3}= \mp {8(5\gamma-4)\sqrt{(\gamma-1)(3-2\gamma)}}$ \\
			\hline
			$\Collins^{\pm}$ & \rule{0pt}{4ex} $ \lambda_{1}=  \pm {\frac {8\,\sqrt {2-\gamma} ( \gamma-1 )  ( 10-9\gamma ) }
				{\sqrt {3\,\gamma-2}}}
			$\\
			&  \rule{0pt}{6ex}  $\lambda_{2} = \pm 2\,{\frac {\sqrt { \left(2 - \gamma \right) }
					\left( \gamma-1 \right)  \left( 6-5\,\gamma+\sqrt {-Q_{3}(\gamma)}
					\right) }{ \sqrt{3\,\gamma-2} }   }$ \\
			& \rule{0pt}{6ex}  $\lambda_{3}= \pm 2\,{\frac {\sqrt {  \left(2 - \gamma \right) }
					\left( \gamma-1 \right)  \left( 6 - 5\,\gamma - \sqrt{-Q_{3}(\gamma)}
					\right) }{\sqrt{3\,\gamma-2}}    }$  \\
			\hline
			$\LK^{\pm}$
			& \rule{0pt}{4ex}  $\lambda_{1}= \pm 8\, ( \gamma-1 )  ( 7\,\gamma-10 ) $ \\ 
			& \rule{0pt}{5ex}  $ \lambda_{2} = \pm 4\, ( \gamma-1 )  (2 - \gamma ) $\\
			& \rule{0pt}{4ex}  $\lambda_{3}= \pm 8\,\frac{ \left( \gamma-1 \right)  Q_{2} (\gamma) }{(2 - \gamma)} $	\\
			\hline 
			$\RT^{\pm}$
			& \rule{0pt}{4ex}  $\lambda_{1}= \pm \frac{8 (2-\gamma) \sqrt{\gamma - 1} \sqrt{3 - 2\gamma}(10-9\gamma)}{ \sqrt{Q_{1}(\gamma) }} $ \\
			& \rule{0pt}{8ex}  $\lambda_{2}= \pm \frac{8 (\gamma - 1)^{\frac{3}{2}} \sqrt{3 - 2\gamma} \Bigl(4 - 3 \gamma  -\sqrt{ Q_{4}(\gamma) } \Bigl) }{  \sqrt{ Q_{1}(\gamma)  }  } $\\
			& \rule{0pt}{8ex}  $\lambda_{3}= \pm \frac{8 (\gamma - 1)^{\frac{3}{2}} \sqrt{3 - 2\gamma} \Bigl( 4 - 3 \gamma + \sqrt{ Q_{4}(\gamma) } \Bigl) }{  \sqrt{ Q_{1}(\gamma)  }  } $  \\ 
			\hline
			Wainwright $( \gamma = \frac{10}{9} )$
			& \rule{0pt}{4ex}  $\lambda_{1} = 0$ \\
			& \rule{0pt}{4ex}  $ \lambda_{2} = {\frac {16}{81}}\,{\it  Y_{1}} \mp {\frac {16}{81}}\,\sqrt { 533 Y^{2}_{1} - 112 }
			$ \\
			& \rule{0pt}{5ex}  $\lambda_{3} =  {\frac {16}{81}}\,{\it  Y_{1}} \pm {\frac {16}{81}}\,\sqrt { 533 Y^{2}_{1} - 112 }
			$ \\
			\hline
		\end{tabular}
	\end{center}
\end{table}

\section{Analysis Of The Vacuum Boundary}

The vacuum boundary is the surface of the elliptical cylinder 
\begin{eqnarray}
Y_{4} = 0 =  16(3-2 \gamma)(\gamma - 1) (1-Y_{1}^{2}) - Y^{2}_{2}, 
\end{eqnarray}
bounded by $S=0$ and $S = 1$, as illustrated in Fig \ref{FigureUnfoldingVacuumBoundary}.\\
In order to display the solution curves, we make a cut at $Y_{1} = -1$ and unwrap the surface as illustrated, again, in Fig \ref{FigureUnfoldingVacuumBoundary}. Note that we identify the two sides of the resulting rectangle when $Y_{1} = -1$.\\
\begin{figure}[h!]
	\caption{Unfolding the Vacuum Boundary ($1 = Y^{2}_{1} + \frac{Y^{2}_{2}}{ 16(\gamma - 1)(3-2 \gamma)}$, on $S$ vs. $Y_{1}$) }
	\centering
	\includegraphics[width=0.7\textwidth]{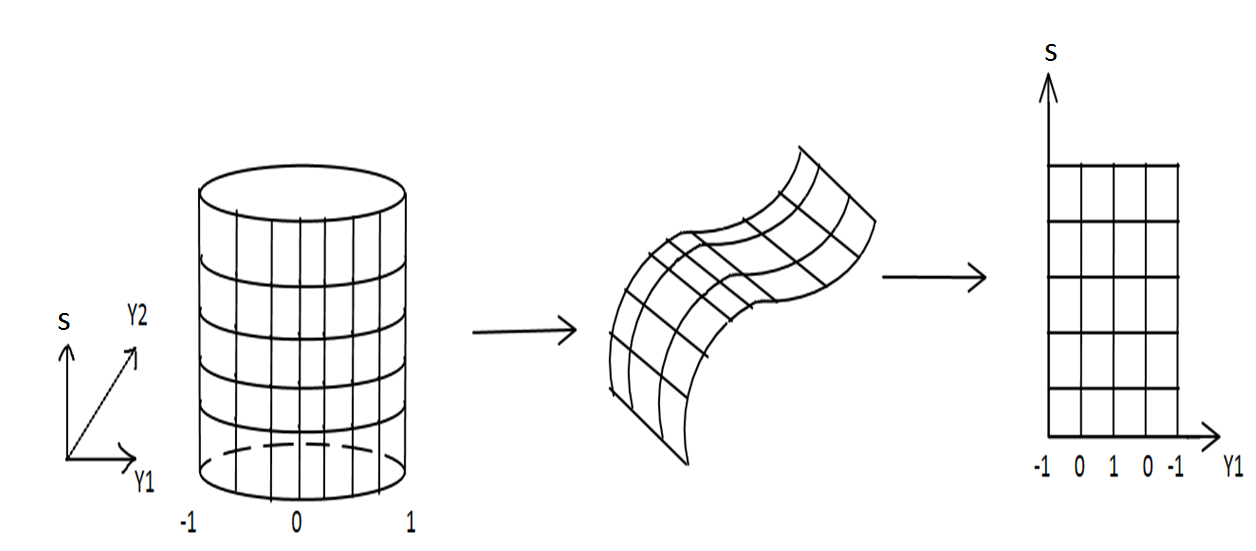} \label{FigureUnfoldingVacuumBoundary}
\end{figure}
When restricted to the vacuum boundary, the DE is two dimensional and we write it as
\begin{eqnarray}
\fl \frac{d Y_{1}}{d \chi} = \frac{4 (\gamma - 1)}{(2 - \gamma)}(1-S^{2})  \ \Bigl[ 4 \gamma ( 3 - 2 \gamma) - Y_{1} L \Bigl] + 4(\gamma - 1)(10 - 7 \gamma) \ ( S^{2} - Y^{2}_{1} ) - Y^{2}_{2}, \label{VacuumBoundarySystem1} \\
\fl \frac{d S}{d \chi} = -\frac{ 4(\gamma - 1)}{(2-\gamma)} S(1-S^{2}) L, \label{VacuumBoundarySystem2} 
\end{eqnarray}
with, $Y^{2}_{2} = 16(\gamma - 1)(3-2 \gamma)(1-Y^{2}_{1})$.\\
The equilibrium points for this system are listed in Table \ref{TableEquilVacuumBoundary1} and the, (signs of the) eigenvalues of the linearization matrix at each equilibrium point are listed in Table \ref{TableofSignofTheEigenvaluesVacuumBoundaryPlus}. These are (the signs of) the second and third eigenvalues from Table \ref{TableofEigenvaluesPlusMinus}.
\begin{table}[h!]
	\caption{Equilibrium points on the vacuum boundary} \label{TableEquilVacuumBoundary1}
	\lineup
	\begin{center}
		\begin{tabular}{@{}*{3}{l}}
			\br                              
			Type of Solution &  $S^{2}$ & $Y_{1}$\\
			\hline
			$\mathrm{INF}^{\pm}$ & $0$ &  $ \pm \frac{\sqrt{4(3 - 2 \gamma)(\gamma - 1)}}{(\gamma - 2)}$  \\ 
			\hline
			$\mathrm{LK}^{\pm}$ & $1$ & $ \pm 1$  \\
			\hline
			$\mathrm{RT}^{\pm}$ & $24 \frac{(\gamma - 1)(3 - 2 \gamma )}{Q_{1}(\gamma)}$ & $ \pm \frac{(4-3\gamma)}{(2 - \gamma)}\sqrt{ \frac{16(\gamma - 1)(3 - 2\gamma)} {Q_{1}(\gamma)} }$ \\
			\hline
		\end{tabular}
	\end{center}	
\end{table}
\begin{table}[h!]
	\caption{Signs of the eigenvalues for the vacuum boundary} \label{TableofSignofTheEigenvaluesVacuumBoundaryPlus}
	\lineup
	\begin{center}
		\begin{tabular}{@{}*{2}{l}}
			\br                              
			Equilibrium Point & $1 < \gamma < \frac{3}{2}$   \\ 
			\hline
			$\mathrm{INF}^{\pm}$ & $\mp \mp$\\ 
			\hline
			$\mathrm{LK}^{\pm}$ & $\pm \pm$\\ 
			\hline
			$\mathrm{RT}^{\pm}$ &  $\mp \pm$\\
			\hline	
		\end{tabular}
	\end{center}
\end{table}

There are only two equilibrium points in the interior of phase space and these are both saddle points. Thus, there cannot be any limit cycles in this system and the phase portrait may be easily deduced by making use of the isoclines. A typical phase portrait is given in Fig \ref{FigureDrawingVacuumBoundary1}.
\begin{figure}[h!]
	\centering
	\includegraphics[width=0.7\textwidth,height=0.3\textheight]{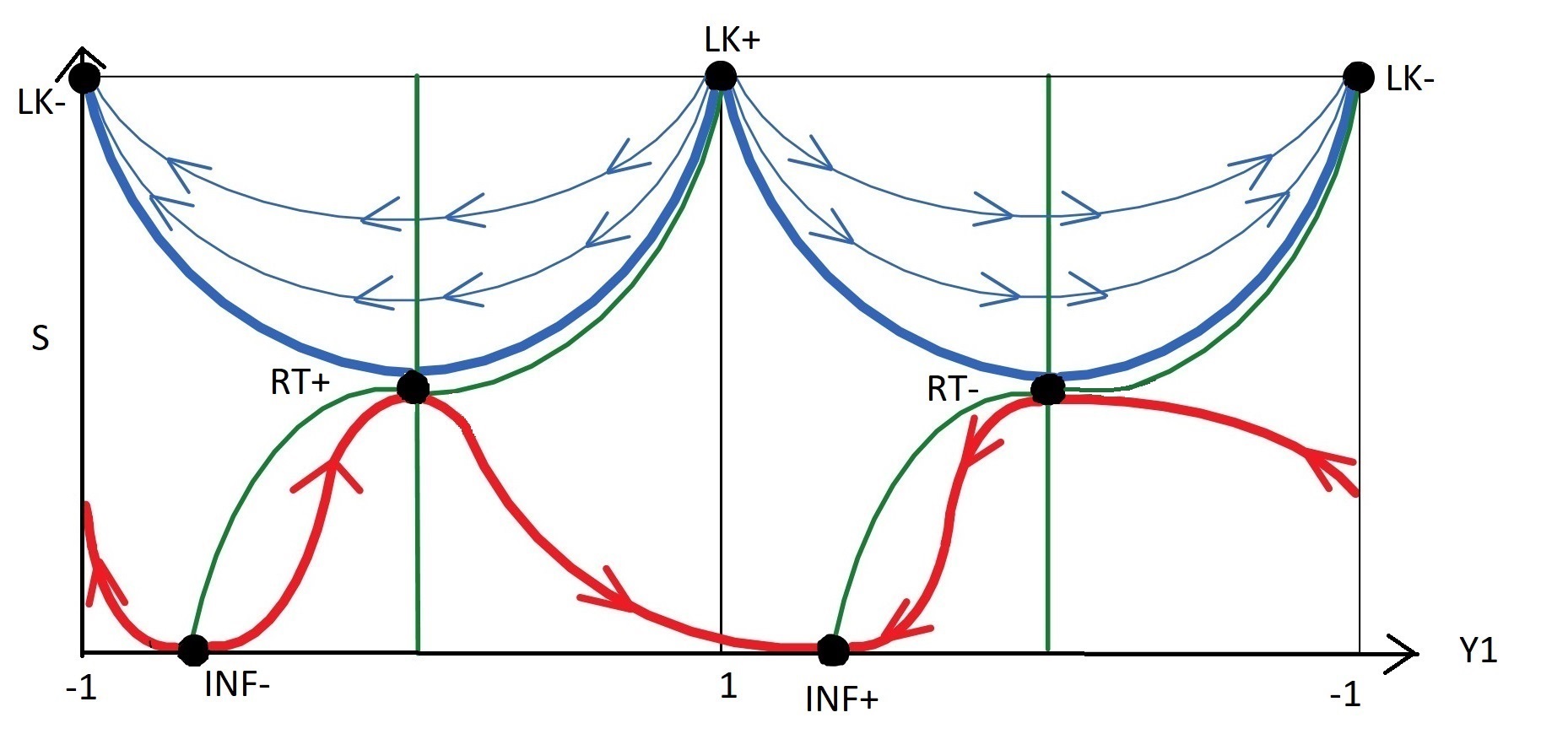}
	\caption{Phase portraits for the vacuum models.}
	\label{FigureDrawingVacuumBoundary1}
\end{figure}

The isoclines $\frac{dS}{d \chi}=0$ are the lines in the phase portraits given by either $S=0$, $S=1$ or the vertical green line $L=0$. The isoclines $\frac{dY_{1}}{d \chi}=0$ are the two green curves passing through the three equilibrium points, INF, RT, and LK. The separatrices of the two RT saddles are indicated in blue and red.

\begin{lemma}
	For $1<\gamma \leq \frac{3}{2}$, there exists an open set of well-behaved vacuum cosmological models in this class with the property that they are asymptotic to LK as $\chi \longrightarrow \pm\infty$.
\end{lemma}
\begin{proof}
It is clear from the phase portrait that there is an open set of trajectories which leave $\LK^{+}$ and terminate at $\LK^{-}$. Each one of the curves corresponds to a vacuum model, within this class, and all of its associated dimensionless variables are bounded. Local asymptotic analysis shows that the expansion is bounded near the equilibrium points and thus the physical variables are also bounded on phase space.
\end{proof}
The line element for the vacuum models in this class was written as a differential equation in \cite{van1988class}, No analysis of the DE was provided. Here we have shown that an open set of these vacuum models are well-behaved and can be considered as inhomogeneous generalization of the LK plane wave models. Asymptotic analysis, similar to that provided in the appendix for the perfect fluid models, reveals that models corresponding to trajectories that either leave or terminate at the points labelled INF have curvature scalars that are unbounded and the models are not well-behaved unless, possibly, when $\gamma = \frac{4}{3}$.
\section{Perfect Fluid Models}

We begin with a result about some of the trajectories in the phase space.

\begin{prop}\label{PropExistenceofOpenSetofTraject}
	\phantom{a}
	\begin{enumerate}[label=(\alph*)]
		\item For $1 < \gamma < \frac{10}{9}$ there exists an open set of trajectories which have $C^{+}$ as a source and $C^{-}$ as a sink.
		\item For $\frac{6}{5} < \gamma < \frac{10}{7}$ there exists an open set of trajectories which have $C^{-}$ as a source and $C^{+}$ as a sink.
		\item For $\frac{10}{7}< \gamma < \frac{3}{2}$ there exits an open set of trajectories which have $LK^{+}$ as a source and $LK^{-}$ as a sink.
	\end{enumerate}
\end{prop}
\begin{proof}
	The work of \cite{hewitt1988qualitativeI} demonstrates the existence of trajectories in the invariant set $S=1$ which have the behaviour stated in the proposition. For the ranges of $\gamma$, stated in the proposition, the sources and sinks in the invariant set $S=1$, extend their stability to the full three-dimensional phase space. We conclude from the approximation property of orbits \cite[p.11]{sibirskiui1975introduction} \cite[p.104]{wainwright2005dynamicalSystems} that there exists an open set of trajectories flowing between the stated equilibrium points.
\end{proof}
We continue with a result about the corresponding cosmological models.
\begin{corollary}\label{CorollaryExistenceOfCosmodels}
	\phantom{a}
	\begin{enumerate}[label=(\alph*)]
		\item For $1 < \gamma < \frac{10}{9}$ and $\frac{6}{5} < \gamma < \frac{10}{7}$ there exists an open set of well-behaved cosmological models within this class that are asymptotically spatially homogeneous and matter dominated as $\chi \longrightarrow \pm \infty$.
		\item For $ \gamma = \frac{10}{7}$ there exists an open set of well-behaved cosmological models within this class that are asymptotically spatially homogeneous and vacuum dominated as $\chi \longrightarrow \pm \infty$.
		\item For $ \frac{10}{7} < \gamma < \frac{3}{2}$ there exits an open set of well-behaved cosmological models within this class that are acceleration dominated and vacuum dominated as $\chi \longrightarrow \pm \infty$.
	\end{enumerate}
\end{corollary}
\begin{proof}
	For the trajectories referred to in Proposition \ref{PropExistenceofOpenSetofTraject}, the variables $Y_{1}$, $Y_{2}$, and $S$ are bounded on phase space, and in addition $\lim\limits_{\chi \to \pm \infty} S = 1$. Since the scaled energy density $Y_{4}$ is quadratic in the variables $Y_{1}$ and $Y_{2}$, then this variable is bounded on phase space, and the spatial part of the energy density, $\Omega(\chi) =  Y_{4}(\chi) S^{2}(\chi)$ is bounded for any of these trajectories. We have already shown in \cite{hewitt2020}, equation (8), that the spatial part of the expansion is determined from the spatial part of the energy density by $\Omega(\chi) = \theta^{\frac{2-\gamma}{\gamma-1}}(\chi)$. It now follows that all of the physical variables, in particular the energy density, acceleration, and shear, for each model associated with any of the trajectories in Proposition \ref{PropExistenceofOpenSetofTraject}, are bounded and so we conclude that the cosmological models are all well-behaved.\\
	Furthermore, the models corresponding to trajectories that are asymptotic to C at large spatial distances, are matter dominated since $\lim\limits_{\chi \to \pm \infty} \Omega = \Omega_{e}$, with $\Omega_{e}$ the value of $\Omega$ at the equilibrium point. In addition $\lim\limits_{\chi \to \infty \pm} \dot{U} =0$ and $\lim\limits_{\chi \to \pm \infty} \dot{U}^{a}_{\ ; a} = 0$\footnote{It is shown in \cite{hewitt1988qualitativeI} that $\frac{\dot{U}^{a}_{\ ; a}}{\theta^{2}} = \frac{1}{6}(3 \gamma - 2)(\Omega - \Omega_{e})$ for these models.}, since these models are approaching a non-vacuum equilibrium point, and so we refer to these models as being asymptotically spatially homogeneous. When $\gamma =\frac{10}{7}$ the C equilibrium points coalesce with the LK equilibrium points, and so the limiting behaviour retains its asymptotic spatially homogeneous feature but is now vacuum dominated since $\lim\limits_{\chi \to \pm \infty} \Omega = 0$. \\
	The trajectories which are asymptotic to LK are vacuum dominated with $\lim\limits_{\chi \to \pm \infty} \Omega = 0$. We also label them as acceleration dominated since $\lim\limits_{\chi \to \pm \infty} \dot{U} \neq 0$ and $\lim\limits_{\chi \to \pm \infty} \dot{U}^{a}_{\ ;a} \neq 0$.
\end{proof}

\subsection{Other Possible Behaviour For Perfect Fluid Models In This Class}

A natural question to ask is whether or not there are other possible behaviours for well-behaved cosmological models in this class. We explain that the answer is no for all values of the equation state parameter $\gamma$, with $1 < \gamma < \frac{3}{2}$, $\gamma \neq \frac{6}{5}$, $\gamma \neq \frac{10}{9}$. These two exceptional values of $\gamma$ are discussed separately. \\
We first introduce an extended version of the Monotonicity Principle \cite[p.103]{wainwright2005dynamicalSystems}. In summary, this result reveals that if we have a function that is strictly monotone on the solution curves of a differential equation within a positive invariant set then the trajectories are future asymptotic ( in the $\omega$-limit set sense) to a particular subset of points on the boundary of this positive invariant set.
\begin{prop}\label{PropExtendedMon}
	Let $S$ be a positive invariant set of a DE on $\Reals^{n}$. Suppose $G: S \longrightarrow \Reals$ is a $C^{1}$ function whose range is $(a,b)$ with $a<b$, $a \in \Reals \cup \{ -\infty \}$ and $b \in \Reals \cup \{ \infty \}$. If $G$ is decreasing on the orbits of the DE within $S$ then $ \forall \mathbf{x} \in S$,	$\omega(\mathbf{x}) \subseteq \{ \mathbf{s} \in \overline{S} \setminus S | \lim\limits_{\mathbf{y} \to \mathbf{s} } G(\mathbf{y}) \neq b  \}$. In addition, if $S$ is an open set, then $\alpha(\mathbf{x}) \cap S= \emptyset $.
\end{prop}
A similar result holds if $G$ is increasing with $\omega(\mathbf{x}) \subseteq \{  s \in \overline{S} \setminus S | \lim\limits_{ \mathbf{y} \to \mathbf{s}} G(\mathbf{y}) \neq a \}$

\begin{proof}
	See proposition A1 in \cite[p.536]{leblanc1995asymptotic}. The only extension here is that we restrict the invariant set of proposition A1 to be positively invariant.
\end{proof}
Proposition \ref{PropExtendedMon} provides information about where the trajectories end up. The proposition also indicates, to some extent, where it came from: if $S$ is open then the trajectories cannot be past asymptotic( in the $\alpha$-limit set sense) to any point within $S$.

We may also extend these results to the case of a negatively invariant set by exchanging $\alpha(x)$ and $\omega(x)$.\\
The application of either of these results require the determination of the monotone function $G$ for the DE (\ref{ODE-Y1Y2D2-1})-(\ref{ODE-Y1Y2D2-3}). The following function fulfills that role. Let
\begin{eqnarray}
	G:=\frac{ Y^{ \frac{81}{40} Q_{2}(\gamma)}_{4}  (1 - S^{2} )^{\frac{81}{40}(2-\gamma)(10-7\gamma)}  }{ S^{ \frac{81}{5} \gamma (3 - 2 \gamma)} }, \label{functionG}
\end{eqnarray}
then,
\begin{eqnarray}
	\frac{d G}{ d \chi} =  - \frac{81}{20} ( 10 - 9\gamma) ( 2 - \gamma)^{2} Y_{2} G. \label{derivativeofG}
\end{eqnarray}
For $1 < \gamma < \frac{3}{2}$, the 2-space $G=c \in \Reals$ has level curves ( $G=c$, $S=S_{0}$, constant ) which are ellipses ($Y_{4}=$ constant). Fig \ref{DiagramDescribingMonotoneGeneralfunction} illustrates some level surfaces of $G$.\\
At this point it is instructive to examine the sign of the sign of the derivative of $Y_{2}$ along the solution curves, when these curves intersect the plane, $P$, given by
\begin{eqnarray}
	P = \{ (Y_{1},Y_{2},S): -1 \leq Y_{1} \leq 1, Y_{2}=0, 0 \leq S \leq 1 \}.
\end{eqnarray}
We have $\frac{dY_{2}}{d\chi}|_{Y_{2}=0}=(12\gamma-4-7\gamma^{2})+(5\gamma-6)(2-\gamma)Y^{2}_{1}+4(\gamma-1)(3\gamma-4)S^{2}$, and we note that $12\gamma-4-7\gamma^{2}$ is positive on $(1,\gamma_{t})$ and negative on $(\gamma_{t},\frac{3}{2})$ with $\gamma_{t} \approx 1.26$.\\
If $1<\gamma\leq \frac{6}{5}$ then $\frac{dY_{2}}{d \chi}|_{Y_{2}=0} >0$ on the plane $P$. Then $P$ acts as a one way membrane for the flow of solution curves, in that if a trajectory crosses this plane then it does so exactly once from negative $Y_{2}$ to positive $Y_{2}$.\\
If $\frac{4}{3} \leq \gamma \leq \frac{3}{2}$ then $\frac{d Y_{2}}{d \chi}|_{Y_{2}=0} < 0$ on the plane $P$, and, for this range of $\gamma$, if a trajectory crosses $P$ then it does so exactly once moving from positive $Y_{2}$ to negative $Y_{2}$.\\
We define the open sets $Z_{+}$ and $Z_{-}$ by:
\begin{eqnarray}
\fl Z_{+} &= \{ (Y_{1},Y_{2}, S) \in \Reals^{3} \ | \ Y_{4} > 0,  Y_{2}>0, \ 0 < S < 1 \   \}, \nonumber \\
\fl Z_{-} &= \{   (Y_{1},Y_{2}, S) \in \Reals^{3} \ | \ Y_{4} > 0,  Y_{2}<0, \ 0 < S < 1 \ \}. \nonumber
\end{eqnarray} 
It follows from the previous discussion that for the range $ 1 <\gamma \leq \frac{6}{5}$, $Z_{+}$ is positively invariant and $Z_{-}$ is negatively invariant, while if $\frac{4}{3} \leq \gamma < \frac{3}{2}$ then $Z_{+}$ is negatively invariant and $Z_{-}$ is positively invariant.\\

\begin{figure}[h!]
	\caption{The surfaces given by $G=c$. The blue surface with plaid pattern describes a larger value of $\beta$ compared to the gray surfaces.}
	\centering
	\includegraphics[width=0.7\textwidth]{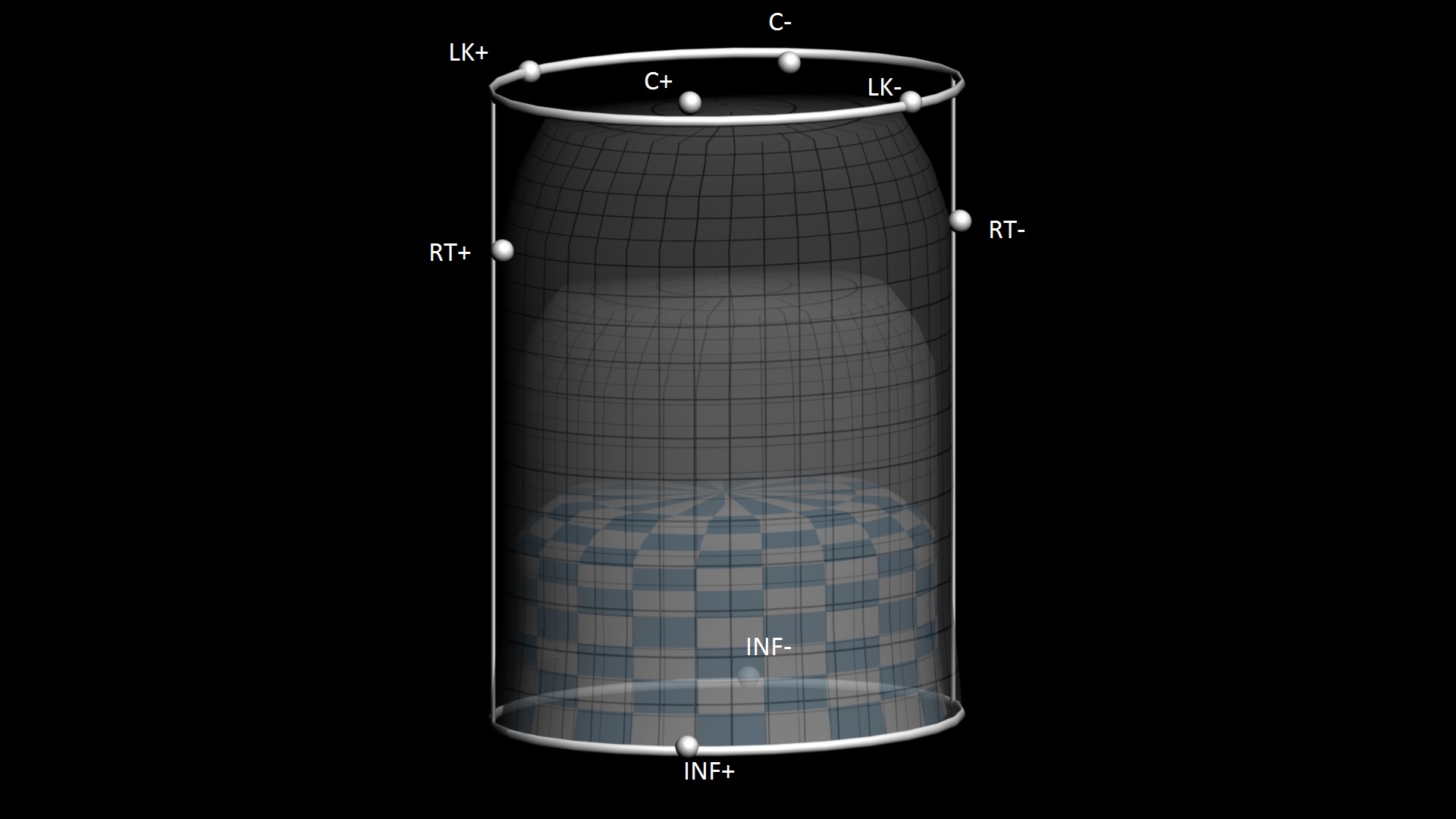}
	\label{DiagramDescribingMonotoneGeneralfunction}
\end{figure}
\begin{prop}\label{TheoremAsymptoticToCollins}
	If $1 < \gamma < \frac{10}{9}$ or $ \frac{4}{3} < \gamma < \frac{10}{7}$ then the only well-behaved perfect fluid cosmological models in this class are asymptotic to the Collins spatially homogeneous model at large spatial distance ($\chi \longrightarrow \pm \infty$). These models are asymptotically spatially homogeneous and matter dominated.
\end{prop}
\begin{proof}
Proposition \ref{PropExistenceofOpenSetofTraject} and Corollary \ref{CorollaryExistenceOfCosmodels} have demonstrated the existence of the models claimed in Proposition \ref{TheoremAsymptoticToCollins}. We need to demonstrate that there are no other possible behaviours for well-behaved models. To do this we make use of the monotone function $G$ on the open sets $Z_{+}$ and $Z_{-}$ together with Proposition \ref{PropExtendedMon}.\\
For $1 < \gamma < \frac{10}{9}$ the function $G$ is monotone decreasing in the positive invariant set $Z_{+}$. If $\mathbf{x} \in Z_{+}$ then $\omega(\mathbf{x}) \subset \overline{Z}_{+} \setminus Z_{+}$, by Proposition \ref{PropExtendedMon}. We now apply the extended Poincar\'e--Bendixson theorem \cite[p.709]{hewitt1991investigation} to conclude that trajectories tend to either $\mathrm{C}^{-}$ or $\mathrm{INF}^{+}$ as $\chi \longrightarrow \infty$.\\
If we consider $\alpha(\mathbf{x})$ then Proposition \ref{PropExtendedMon} shows that $\alpha(\mathbf{x}) \not\subseteq Z_{+}$ and then either $\alpha(\mathbf{x}) \subset \overline{Z}_{+} \setminus Z_{+}$ or $\alpha(\mathbf{x}) \subset \overline{Z}^{c}_{+}$. The first possibility can be eliminated as there are no sources on the boundary of $Z_{+}$. We conclude that if we follow the trajectory through $\mathbf{x}$ backwards along the curve (i.e. decreasing $\chi$) then it came from $Z_{-}$. Since $Z_{-}$ is negatively invariant then Proposition \ref{PropExtendedMon} yields $\alpha(\mathbf{x}) \subset \overline{Z}_{-} \setminus Z_{-}$ and by the extended Poincar\'e--Bendixson theorem \cite[p.709]{hewitt1991investigation} either $\alpha(\mathbf{x}) = \mathrm{C}^{+}$ or $\alpha(\mathbf{x}) = \INF^{+}$.\\
We prove in the appendix that any trajectory in the phase space that tends to $\INF^{+}$ and corresponds to a perfect fluid model has the property that the matter energy density, $\mu$, diverges on a spatial slice, $t=t_{0}>0$. Thus we conclude that the only well-behaved models are the ones corresponding to trajectories tending to $\mathrm{C}^{\pm}$ at large spatial distance. \\
The proof is similar for the case of $\frac{4}{3} < \gamma < \frac{10}{7}$, with the roles of $Z_{+}$ and $Z_{-}$ switching.
\end{proof}
\begin{figure}[h!]
	\caption{Illustration of a trajectory connecting $\Collins^{+}$ to $\Collins^{-}$.}
	\centering
	\includegraphics[width=0.9\textwidth,height=0.3\textheight]{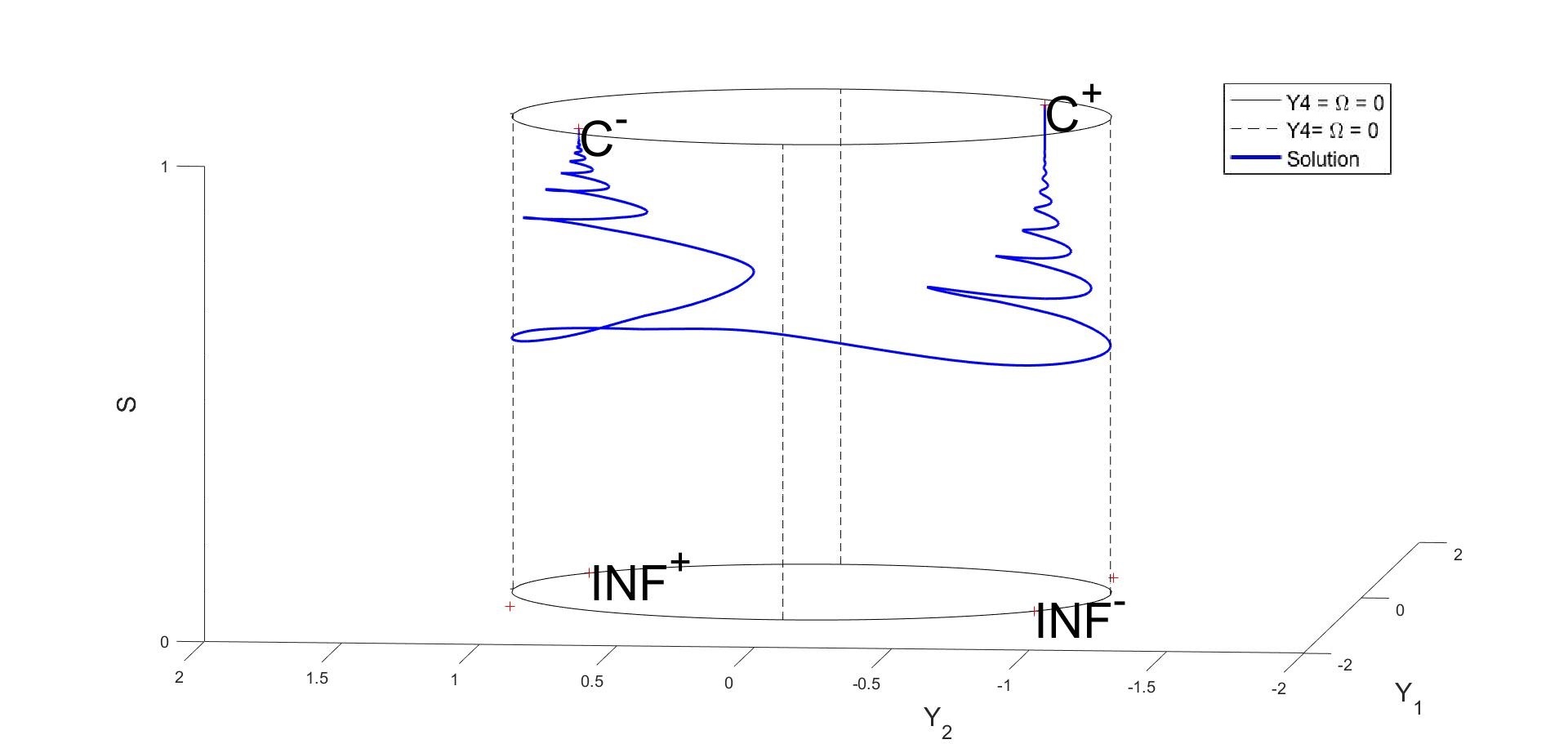}
	\label{Figure3DCollinsCollins}
\end{figure}
\begin{figure}[h!]
	\caption{ $\Omega$ vs. $\chi$ for the solution in Fig \ref{Figure3DCollinsCollins}.}
	\centering
	\includegraphics[width=0.6\textwidth,height=0.3\textheight]{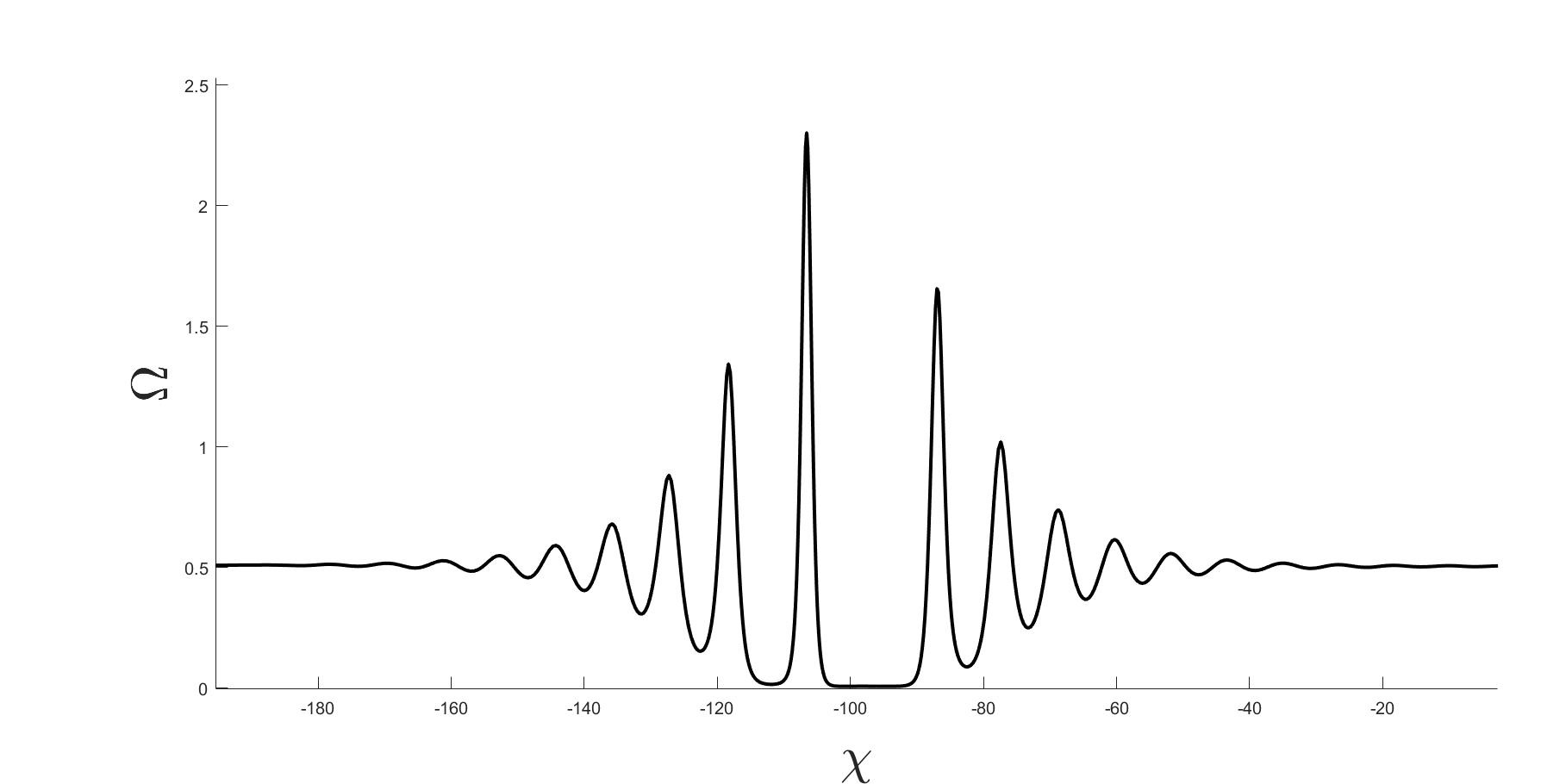}
	\label{FigureOmegaVSChi}
\end{figure}
A typical solution of this type is given in Fig \ref{Figure3DCollinsCollins}. The behaviour of the dimensionless energy density $\Omega$ in terms of the spatial parameter $\chi$ is of our interest. Fig \ref{FigureOmegaVSChi} illustrates $\Omega(\chi)$ for the solution in Fig \ref{Figure3DCollinsCollins}. 
\begin{prop}\label{TheoremAsymptoticToLK}
	If $\frac{10}{7} \leq \gamma < \frac{3}{2}$ then the only well-behaved perfect fluid cosmological models in this class are asymptotic to the LK model at large spatial distance. The models are vacuum dominated, in that $\lim\limits_{\chi \to \pm \infty} \Omega = 0$. If $\gamma = \frac{10}{7}$ then the models are asymptotically spatially homogeneous with $\lim\limits_{\chi \to \pm \infty} \dot{U} = 0 = \lim\limits_{\chi \to \pm \infty} \dot{U}^{a}_{\ ; a}$. If $\frac{10}{7} < \gamma < \frac{3}{2}$ then the models are acceleration dominated with $\lim\limits_{\chi \to \pm \infty} \dot{U} \neq 0$,  $\lim\limits_{\chi \to \pm \infty} \dot{U}^{a}_{\ ; a} \neq 0$.
\end{prop} 
\begin{proof}
Proposition \ref{PropExistenceofOpenSetofTraject} and Corollary \ref{CorollaryExistenceOfCosmodels} demonstrate the existence of such models. The rest of the proof follows the proof of Proposition \ref{TheoremAsymptoticToCollins}.
\end{proof}
A typical solution of this type is given in Fig \ref{Figure3DLKLK}. The dimensionless energy density, $\Omega(\chi)$, of the solution in Fig \ref{Figure3DLKLK} is given in Fig \ref{FigureOMEgaLKLK}.
\begin{figure}[h!]
	\caption{Illustration of a trajectory connecting $\mathrm{LK}^{+}$ to $\mathrm{LK}^{-}$.}
	\centering
	\includegraphics[width=0.75\textwidth,height=0.3\textheight]{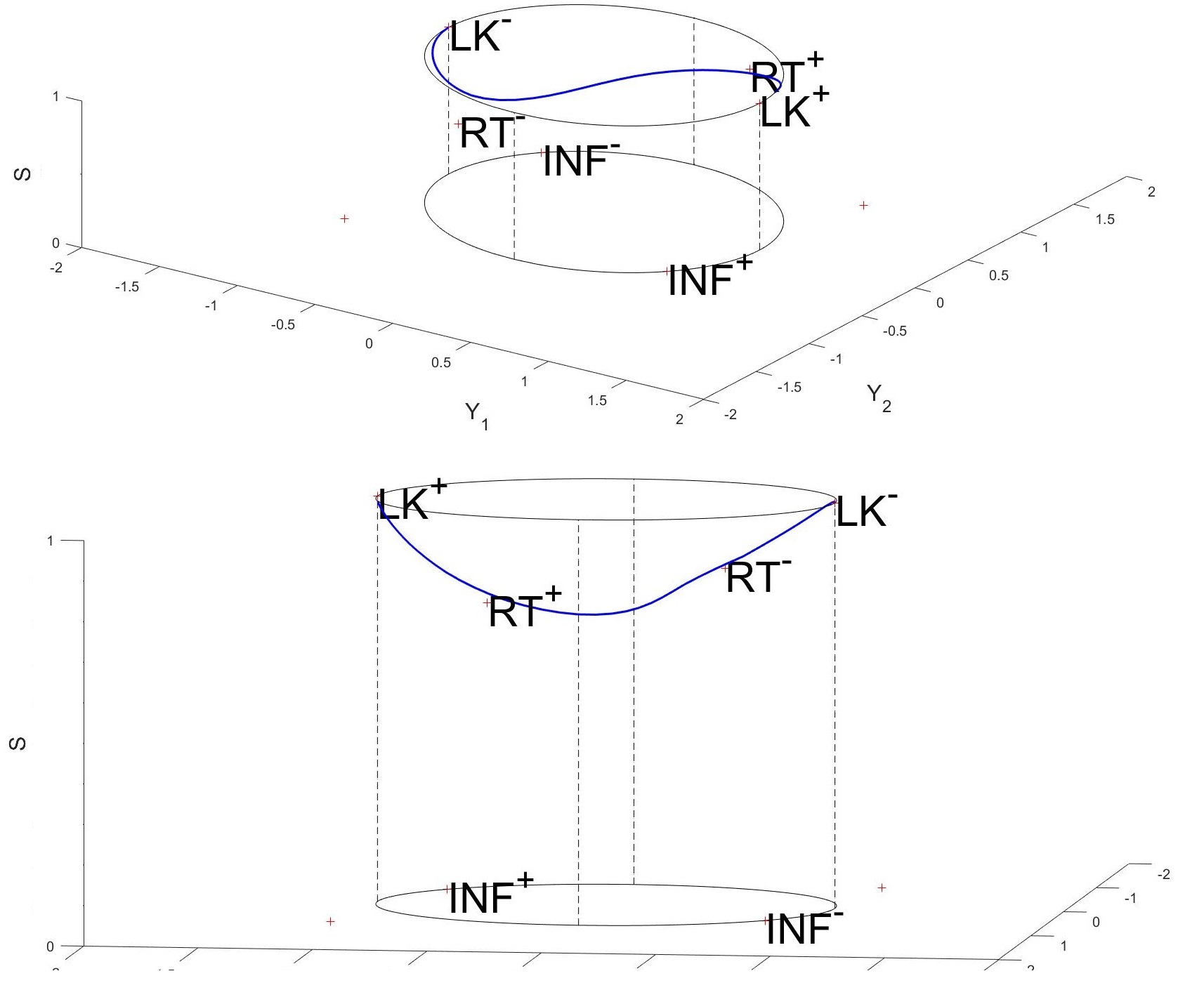}
	\label{Figure3DLKLK}
\end{figure}
 \\
\begin{figure}[h!]
	\caption{$\Omega$ vs. $\chi$,  $\mathrm{LK}^{+}$ to $\mathrm{LK}^{-}$.}
	\centering
	\includegraphics[width=0.8\textwidth]{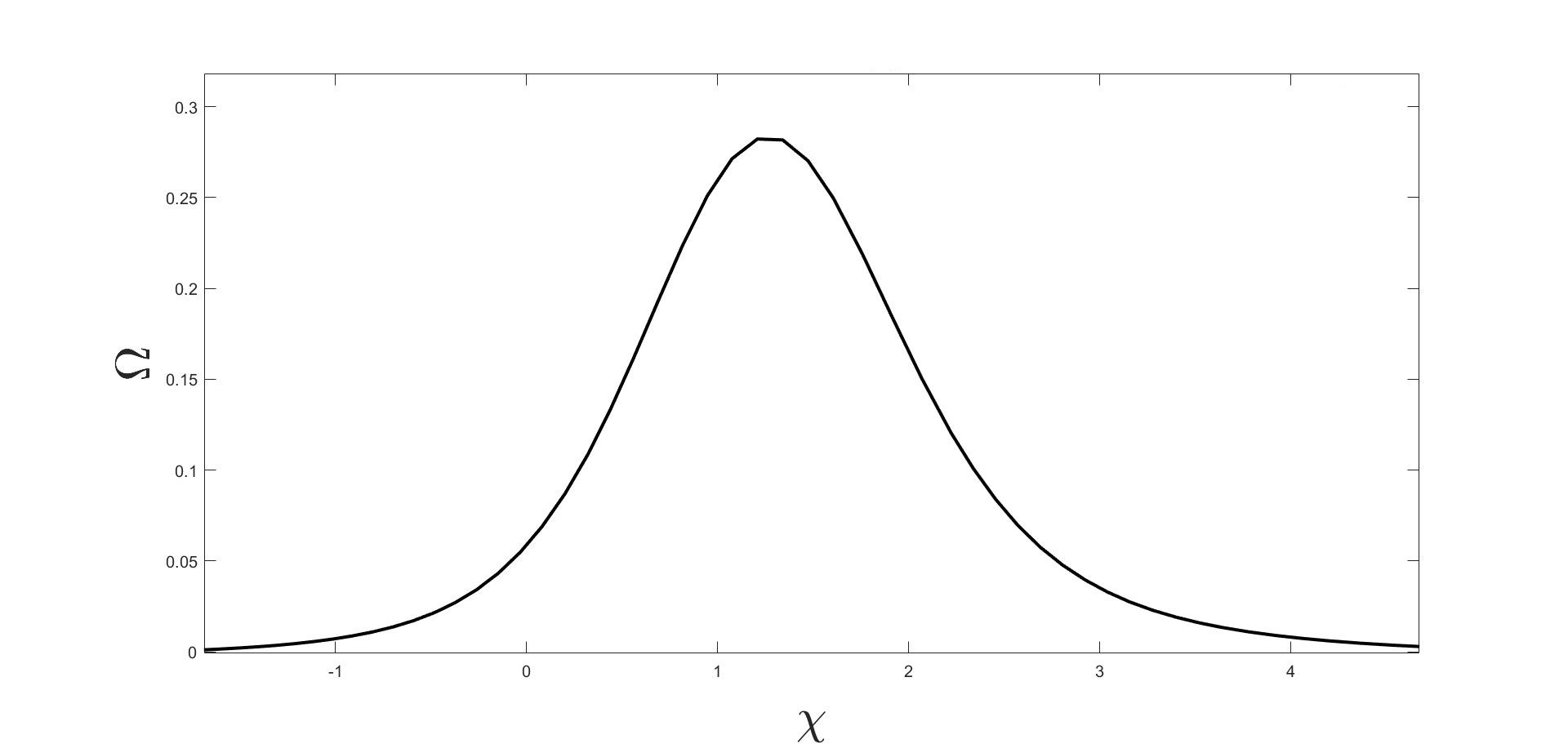}
	\label{FigureOMEgaLKLK}
\end{figure}
We consider the special value of $\gamma$, $\gamma = \frac{10}{9}$, in the next section. We conclude this section with a brief discussion about the case of $\gamma = \frac{6}{5}$.\\
When $\gamma = \frac{6}{5}$, the trajectories corresponding to perfect fluid models in the invariant set $S=1$ are closed curves. To the best of the authors' knowledge these models were the first exact solutions of the EFEs to be found in which the spatial inhomogeneity is periodic ( see \cite{hewitt1988qualitativeI} ).\\
\begin{figure}[h!]
	\begin{center}
    \centering
	\caption{The invariant set $S=1$, $Y_{2} \geq 0 $, when $\gamma = \frac{6}{5}$. The phase portraits consist of closed curves and a cycle graph on the vacuum boundary which asymptotically tends to $\mathrm{LK}^{+}$ and $\mathrm{LK}^{-}$ equilibrium points as illustrated in \cite[p.1320]{hewitt1988qualitativeI}. The blue surface represents a larger value of $c$ compared to the gray surfaces. }
	\centering
	\includegraphics[width=0.75\textwidth]{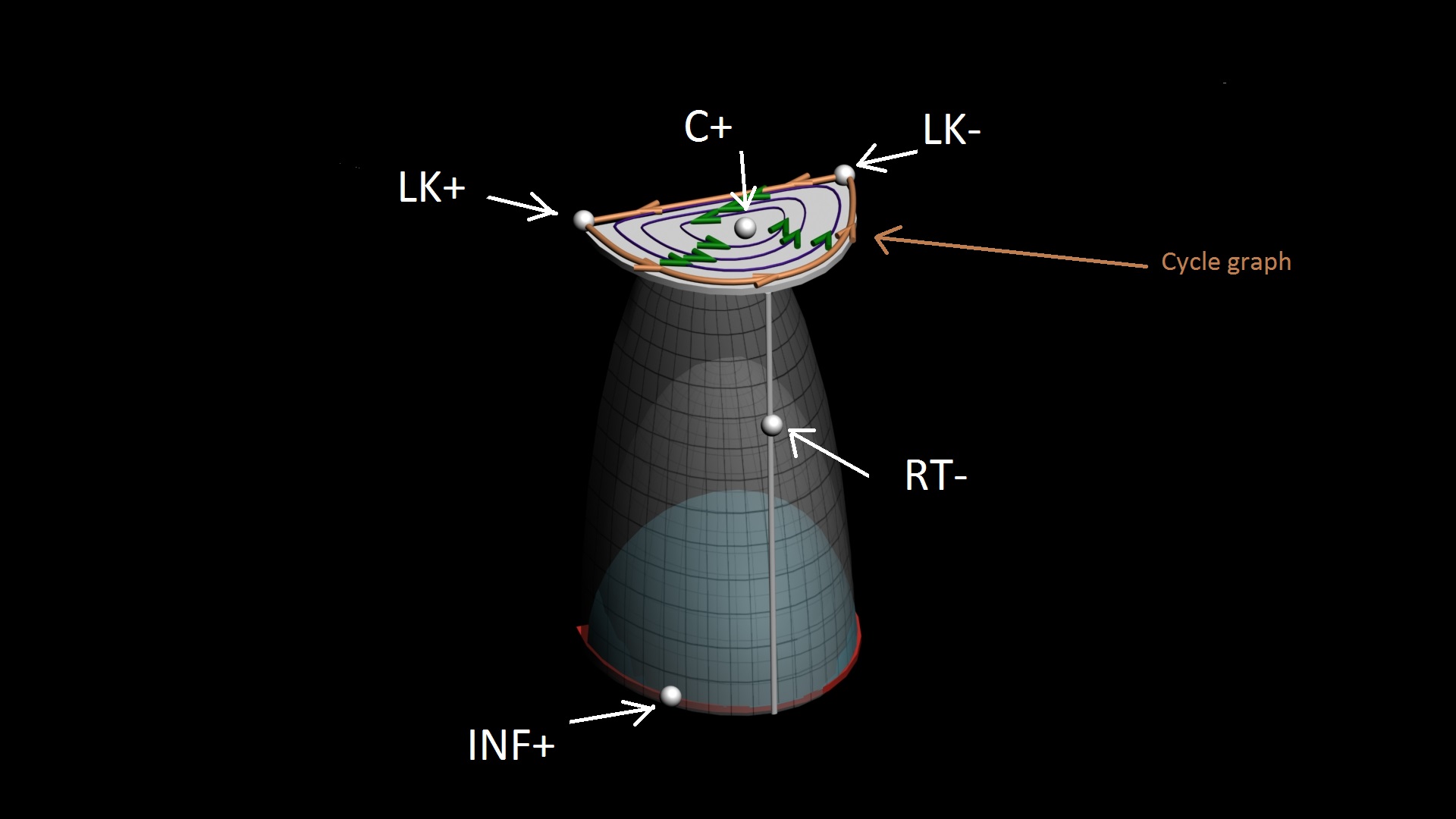}
	\label{DiagramDescribingMonotoneincreasing function}
	\end{center}
\end{figure}
The function $G$ is monotone increasing in the positively invariant set $Z_{+}$ and then trajectories are asymptotic to the boundary of $Z_{+}$ as $\chi \longrightarrow +\infty$. Following the trajectories into the past eventually forces us into $Z_{-}$, in which $G$ is monotone decreasing, and are thus asymptotic to the boundary of $Z_{-}$ as $\chi \longrightarrow - \infty$. Applying the extended Poincar\'e--Bendixson theorem \cite[p.709]{hewitt1991investigation} leads to the possibilities presented in Proposition \ref{Propgamma65}.
\begin{prop}\label{Propgamma65}
If $\gamma = \frac{6}{5}$, and if there exist trajectories that correspond to well-behaved perfect fluid models, then these trajectories are asymptotic to one of the following at large spatial distance as $\chi \longrightarrow +\infty$ and as $\chi \longrightarrow - \infty$.
\begin{enumerate}[label=(\alph*)]
	\item The C equilibrium point.
	\item One of the closed curves in $S=1$.
	\item The cycle graph in $S=1$ connecting the two LK equilibrium points.
\end{enumerate}
\end{prop}
\begin{proof}
	The proof of the result is analogous to the proof of Proposition \ref{TheoremAsymptoticToCollins}.
\end{proof}
Note that in this case we are unable to demonstrate, analytically, the existence of any trajectories corresponding to well-behaved models.\\

\section{Perfect Fluid Models With $\gamma = \frac{10}{9}$ }
If the equation of state parameter, $\gamma$, is equal to $\frac{10}{9}$, then the function $G$ defined in (\ref{functionG}) has a derivative of zero. Thus, the DE has a first integral, and all trajectories lie on a surface $G =$constant, that is
\begin{eqnarray}
	\frac{Y^{3}_{4}(1 - S^{2})^{4}}{S^{14}} = \text{constant}. \ ( \geq 0 ) 
\end{eqnarray}
The surface $G=0$ is the union of the two surfaces that have already been considered, namely the vacuum boundary $Y_{4}=0$, and the "top" of the phase space, the invariant set $S=1$. A phase portrait lying on this union of two surfaces is given in Fig \ref{FigureBifurcationGSteps1}. The surface $G = $"$\infty$" is the "bottom" of phase space, the invariant $S=0$ which corresponds to trajectories at infinity. A phase portrait lying on this surface is given in Fig \ref{FigureBifurcationGSteps5}. Surfaces of the form $G=c$, $c \in \Reals^{+}$ have been briefly discussed in the previous section and illustrated in Fig \ref{DiagramDescribingMonotoneGeneralfunction}. Note that all of these surfaces intersect on the curve $S = 0$, $Y_{4}=0$.\\
Furthermore, when the equation of state parameter, $\gamma$, is equal to $\frac{10}{9}$, the DE admits the one parameter family of equilibrium points, labeled W\footnote{The exact transitively self-similar cosmological models which corresponds to these equilibrium points was first given by Wainwright in \cite{hsu1986self}.}. These equilibrium points form an ellipse in phase space, however we are only concerned with the two segments of this ellipse that lie in the physical region of phase space. Each one of these segments connects one of the vacuum RT equilibrium points to one of the C equilibrium points and enables an exchange of stability between these two disconnected equilibrium points. These segments are illustrated in Fig \ref{CurvesOfWainwrightEquiPoints}.
\begin{figure}[h!]
	\caption{Curves of Wainwright equilibrium points connecting RT and C equilibrium points.}
	\centering
	\includegraphics[width=0.75\textwidth]{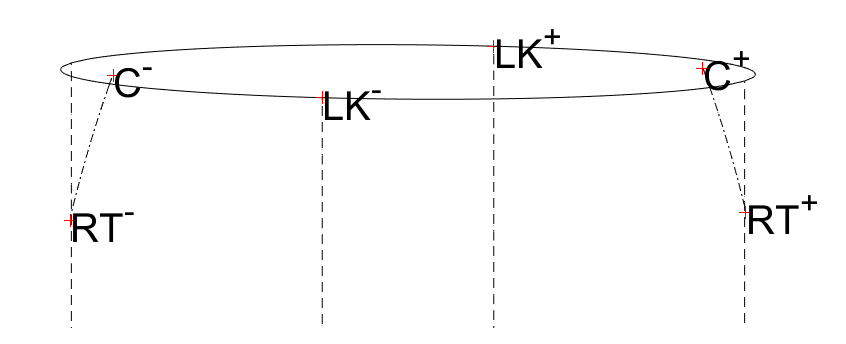}
	\label{CurvesOfWainwrightEquiPoints}
\end{figure}
This type of bifurcation has been identified in numerous cosmological dynamical systems and was first detected in \cite{hewitt1991investigation}, in which this identical curve of equilibrium points is responsible for a change of stability between the identical terminal equilibrium points (namely RT and C). However, in that example the effect concerns the asymptotic state at large times rather than at large distance, as is the case here. We refer to this bifurcation as a line bifurcation. It is an unusual bifurcation because it allows a transfer of stability between two equilibrium points in phase space that never come in contact with each other.\\
The surfaces $G=c$ intersect the two segments of W equilibrium points in either zero, two, or four equilibrium points depending on whether $c>c^{*}$, $c=c^{*}$ or $ 0 < c < c^{*}$, where $c^{*} = \frac{2^{8} 5^{7}}{7^{11}3^3} \approx 3.75 \times 10^{-4}$, respectively.\\
When there are two intersections the equilibrium points are both saddle-nodes. When there are four intersections, then two of the equilibrium points are saddle points and the other two are either foci, $0 < c <c^{**}$, or nodes, $c^{**} \leq c \leq c^{*}$, where $c^{**} = \frac{2^{8} 17^{4} 41^{3}}{13^{14}} \approx 3.74 \times 10^{-4}$.\\
There are four qualitatively distinct types of phase portraits and they are illustrated in Fig \ref{FigureBifurcationGSteps2}-\ref{FigureBifurcationGSteps4}. \\
In order to be able to fully understand how the surfaces $G=c$ mesh together, we also provide the phase portraits on $G=0$, in Fig \ref{FigureBifurcationGSteps1} and on $G="\infty"$, in Fig \ref{FigureBifurcationGSteps5}.\\
\begin{lemma}\label{LemmaExistenceGamma19}
	If $\gamma = \frac{10}{9}$ then there exists trajectories which are both asymptotic to a W equilibrium point as $\chi \longrightarrow \pm \infty$.
\end{lemma}
\begin{proof}
The first part of the proof requires the demonstration that there are no limit-cycles on the surfaces $G=$constant. We consider the smooth vector field $\mathbf{W}=S Y^{-1}_{4} Y^{-1}_{2}(- \frac{dY_{2}}{d \chi}, \frac{d Y_{1}}{d \chi}, 0)$ which has the property that $\nabla \times \mathbf{W} .d\mathbf{S}$ (here $d \mathbf{S}$ is the vector element of surface area) is single-signed on a surfaces $G=$constant within each invariant set $Z_{+}$ and $Z_{-}$. It follows that the hypotheses of the Stokes--Bendixson--Dulac theorem \cite{hewitt1991investigation} are satisfied and thus we conclude that there are no limit cycles in the interior of the phase space\footnote{An equivalent formulation of the non-existence of limit cycles is that if the DE is projected onto the plane $S=0$ then $B$ is a Dulac function for the implicit ($S$ is determined from $G=c$) DE in the plane.}.\\
Secondly, we make use of the approximation property of orbits.\\
There are trajectories flowing from $\mathrm{C}^{+}$ to $\mathrm{C}^{-}$ in the invariant 2-space $S=1$. It follows that in the three-dimensional phase space, there are shadowing orbits with $S\neq1$. Since there are no limit cycles, the shadowing orbits are asymptotic to W equilibrium points, as $G=$constant along the orbits, as $\chi \longrightarrow \pm \infty$. There are trajectories from RT to INF in the vacuum boundary ($Y_{4}=0$). It follows that there are shadowing orbits with $Y_{4}\neq0$. We are also guaranteed that there cannot be any saddle connections between W equilibrium points. This follows due to the fact that the symmetry in the DE only allows such saddle connections to arise in pairs, thus creating a closed connected path which would be the union of trajectories, and this is excluded by the Stokes--Bendixson--Dulac theorem.
\end{proof}
The phase portraits for the models which flow from W to W, that is for $G=c$, $0 < c < c^*$, is given in Fig \ref{FigureBifurcationGSteps2}. We remark that the possibility of any saddle connections between the W saddle points may be eliminated by using the invariant properties of $Z_{+}$ and $Z_{-}$.\\
When $G=c^{*}$, the two W foci have become nodes and the plane portraits for $G=c$, $c^* \leq c < c^{**}$ are given in Fig \ref{FigureBifurcationGSteps2-2}. When $c = c^{*} = c^{**}$ the two pairs of equilibrium points undergo a saddle-node bifurcation \cite[p.177]{hirsch2012differential}.\\
For  $c > c^{*}$ there are no equilibrium points in the interior of the phase space and so trajectories flow from $\mathrm{INF}^{-}$ to $\mathrm{INF}^{+}$. There is no other possibility. \\
In the next result we provide a cosmological interpretation for well-behaved models.\\
\begin{figure}[h!]
	\caption{$G=0$, this is the union of the phase portraits Fig \ref{FigureDrawingVacuumBoundary1} ( $Y_{4}=0$ ) and \cite[Figure 2]{hewitt1988qualitativeI} ( S = 1). }
	\centering
	\includegraphics[width=0.85\textwidth,height=0.2\textheight]{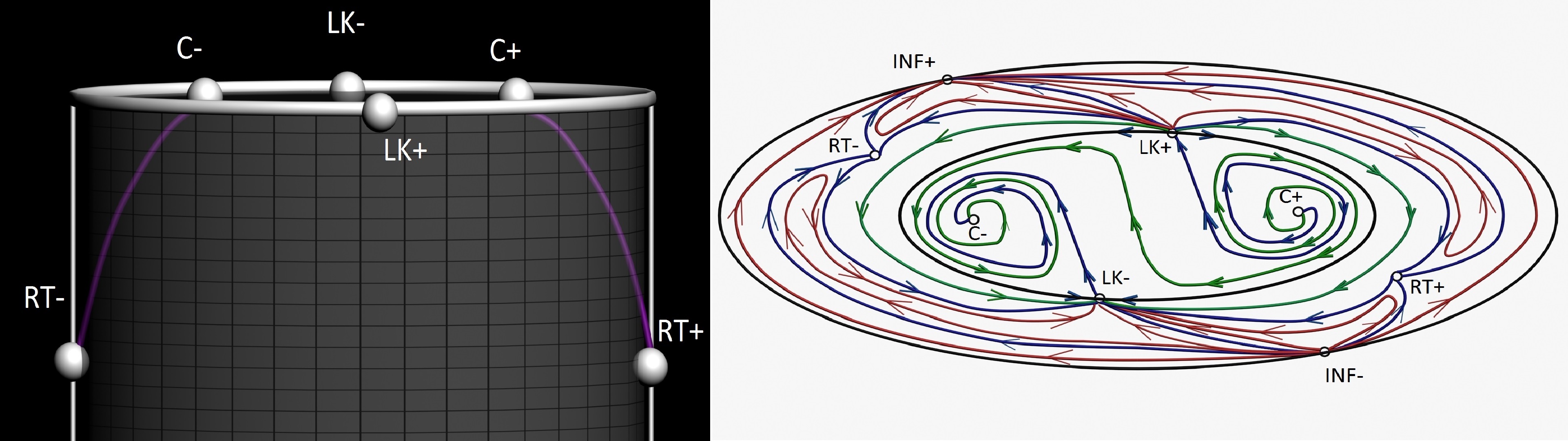}
	\label{FigureBifurcationGSteps1}
\end{figure}
\begin{figure}[h!]
	\caption{$ 0 < G < c^{**}$, the surface intersects with four Wainwright equilibrium points.}
	\centering
	\includegraphics[width=0.85\textwidth,height=0.2\textheight]{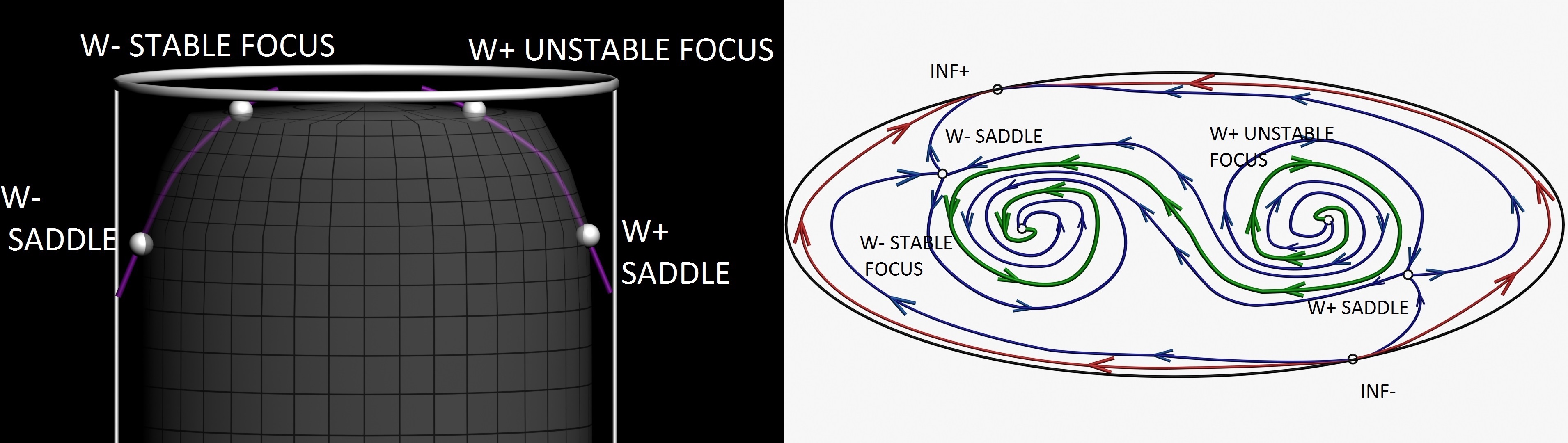}
	\label{FigureBifurcationGSteps2}
\end{figure}
\begin{figure}[h!]
	\caption{$c^* \leq c < c^{**}$, the surface intersects with four Wainwright equilibrium points.}
	\centering
	\includegraphics[width=0.85\textwidth,height=0.2\textheight]{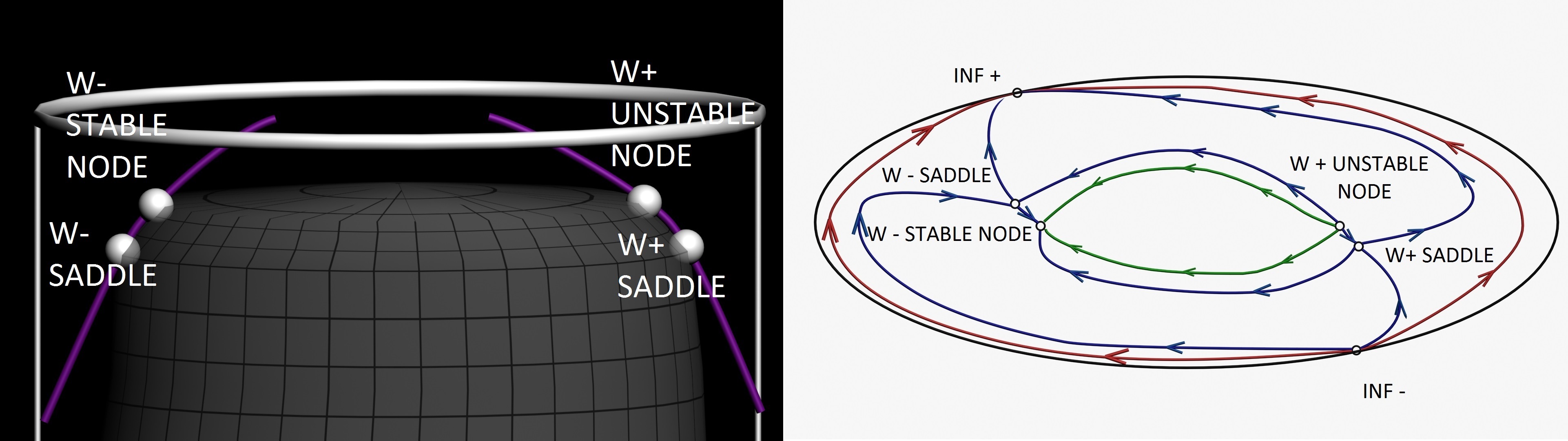}
	\label{FigureBifurcationGSteps2-2}
\end{figure}
\begin{figure}[h!]
	\caption{The invariant 2-space $G=c^*$.}
	\centering
	\includegraphics[width=0.85\textwidth,height=0.2\textheight]{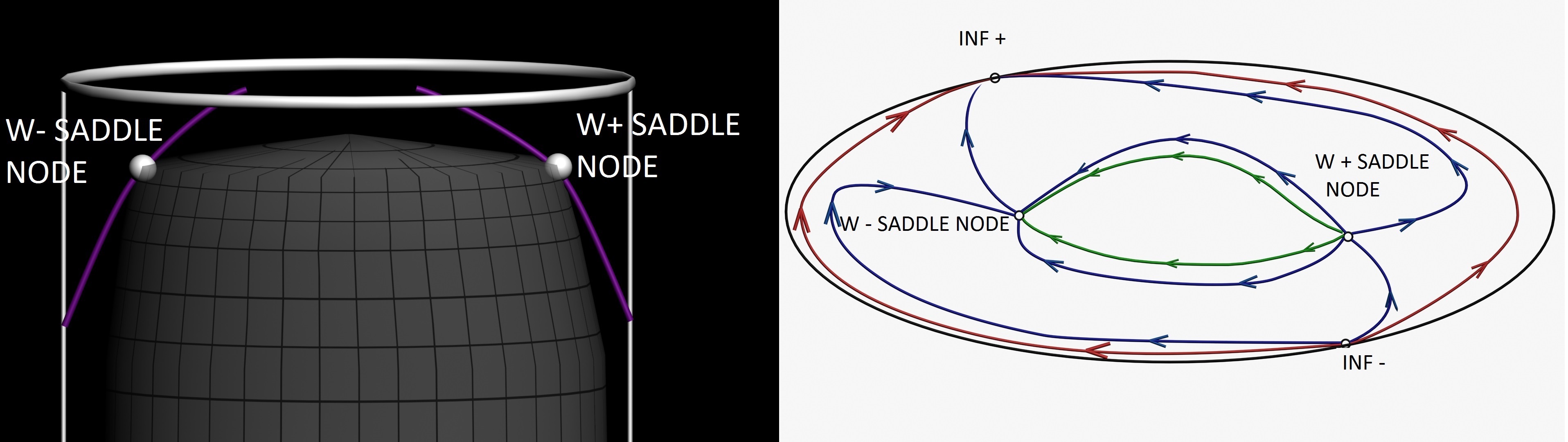}
	\label{FigureBifurcationGSteps3}
\end{figure}
\begin{figure}[h!]
	\caption{The invariant 2-space $G > c^{**}$.}
	\centering
	\includegraphics[width=0.85\textwidth,height=0.18\textheight]{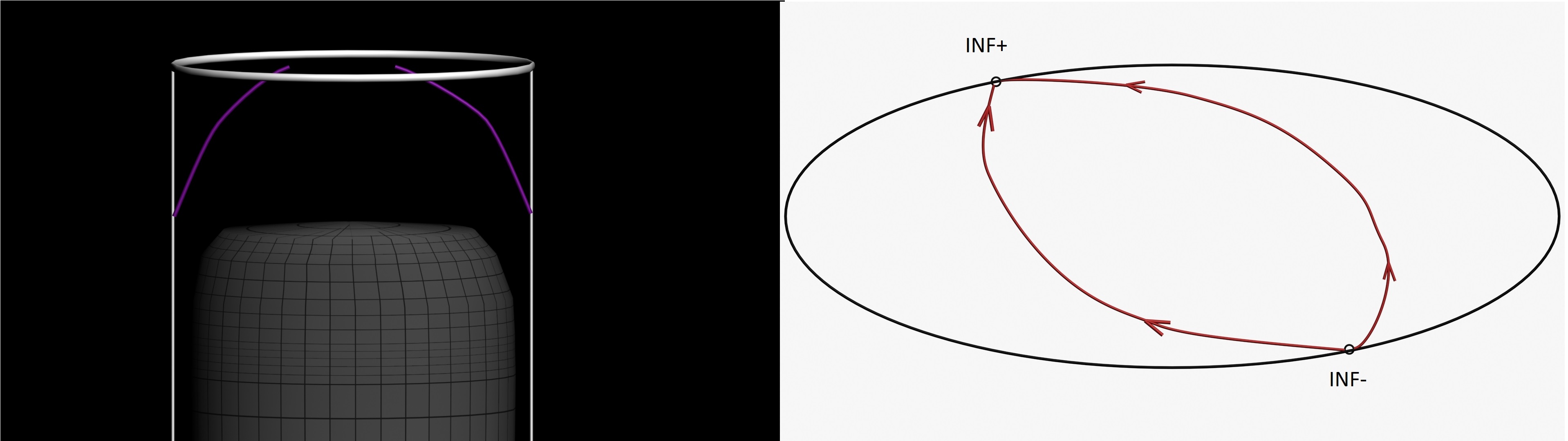}
	\label{FigureBifurcationGSteps4}
\end{figure}
\begin{figure}[h!]
	\caption{The invariant 2-space $S=0$.}
	\centering
	\includegraphics[width=0.85\textwidth,,height=0.18\textheight]{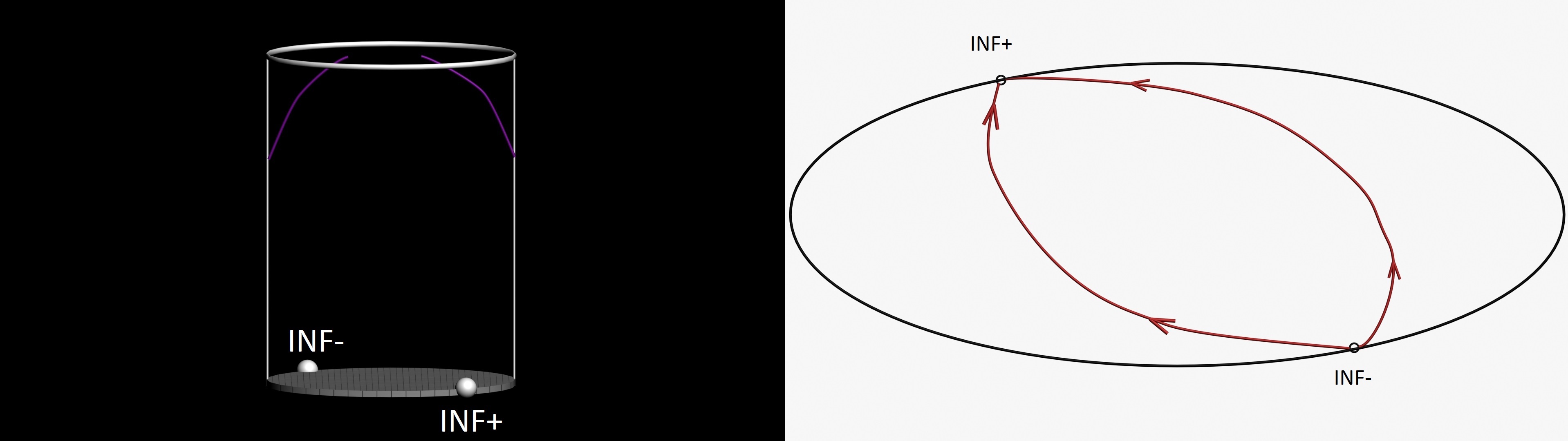}
	\label{FigureBifurcationGSteps5}
\end{figure}
\newpage
\begin{corollary}
	When the equation of state parameter is equal to $\frac{10}{9}$, there exists an open set of well-behaved cosmological models in this class. The models are matter dominated, in that, $\lim\limits_{\chi \to \pm \infty} \Omega = \Omega_{e} \neq 0$, where $\Omega_{e}$ is the value of $\Omega$ at the equilibrium point. In addition, the models are asymptotically spatially homogeneous, in that $\lim\limits_{ \chi \to \pm \infty} \dot{U} = 0$, $\lim\limits_{\chi \to \pm \infty} \dot{U}^{a}_{\ ; a} = 0$.
\end{corollary}
\begin{proof}
	The proof follows from the proof of Lemma \ref{LemmaExistenceGamma19}.
\end{proof}
\begin{figure}[H]
	\caption{A typical trajectory for $ 0 < c < G_{c}$.}
	\centering
	\includegraphics[width=0.75\textwidth,height=0.3\textheight]{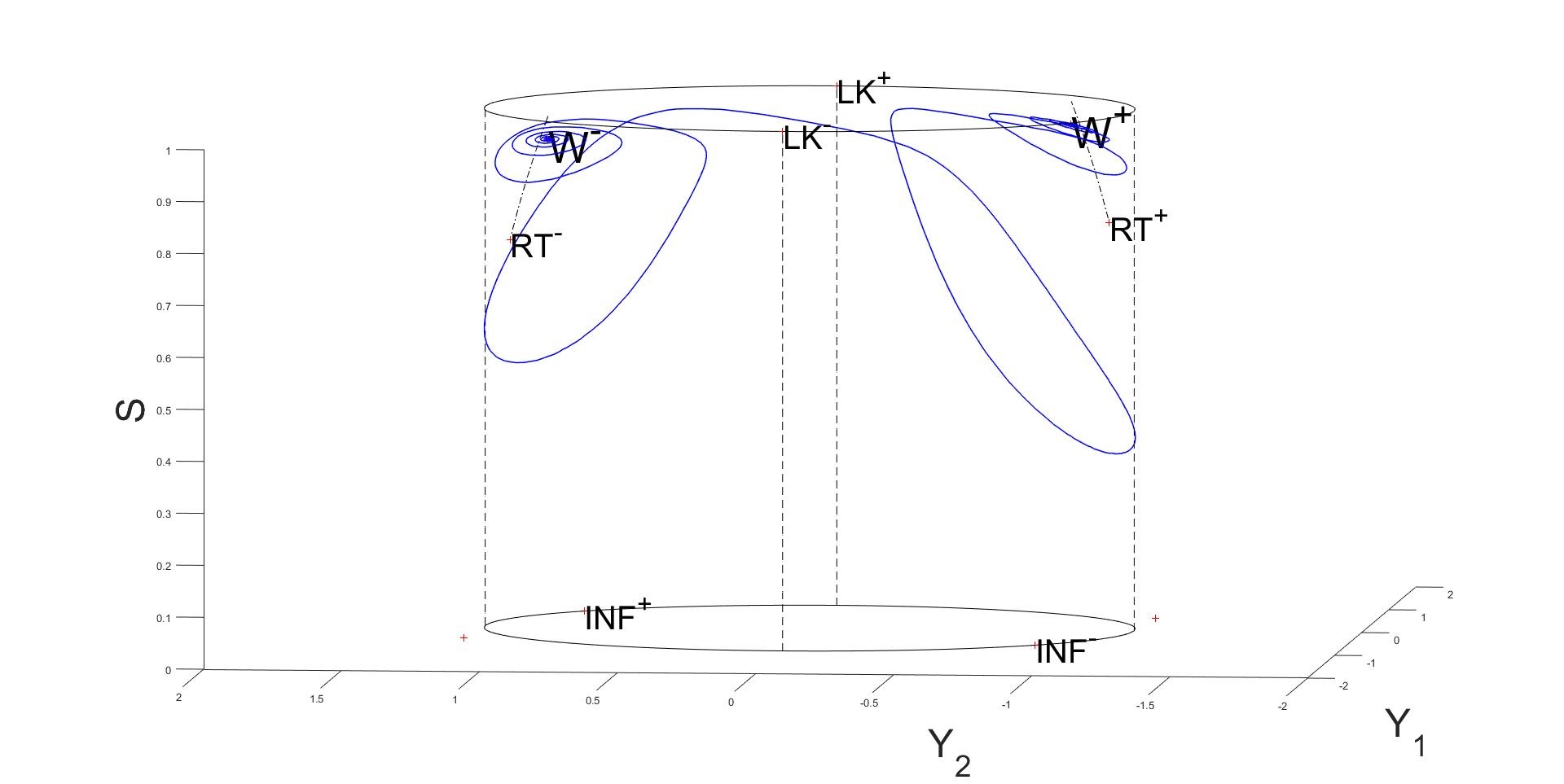}
	\label{FigureBifurcation3D109thPhaseportraits}
\end{figure}
\begin{figure}[H]
	\caption{ $\Omega$ vs. $\chi$,  for the solution in Fig \ref{FigureBifurcation3D109thPhaseportraits}}
	\centering
	\includegraphics[width=0.75\textwidth]{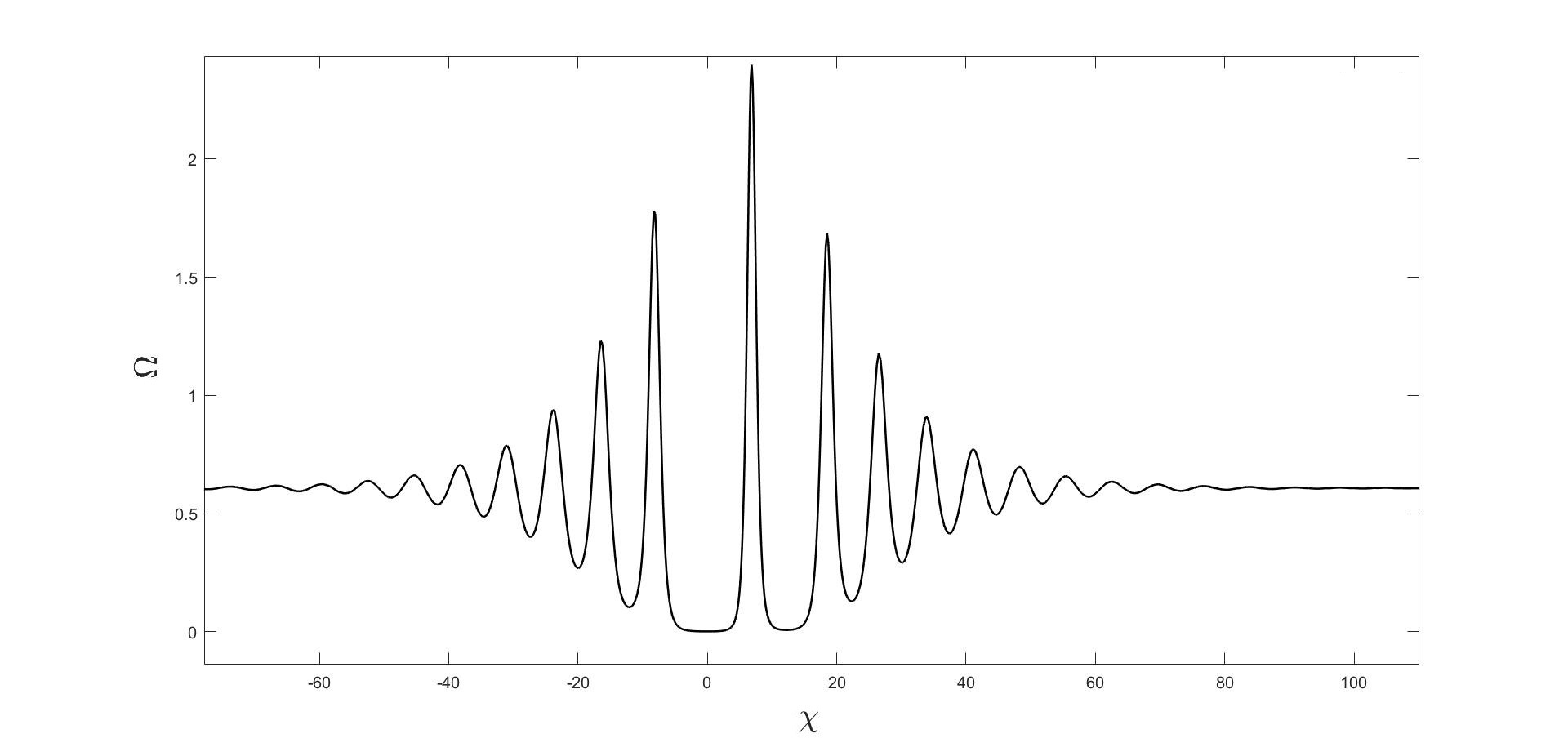}
	\label{FigureBifurcationG4OMEGA2}
\end{figure}
A numerical simulation of a trajectory which is asymptotic the W equilibrium points is given in Fig \ref{FigureBifurcation3D109thPhaseportraits} and the density profile of the corresponding cosmological model is given in Fig \ref{FigureBifurcationG4OMEGA2}.
\phantom{A}\\

\section*{Summary And Conclusion:}
In this paper we have examined a large class of spatially inhomogeneous self-similar cosmological models. They form an exceptional class, in a similar manner to the exceptional Bianchi B class ($VI_{h=-1/9}$) in that one of the KVFs is H.O.\\
We have shown that there exists an open set of well-behaved inhomogeneous models that posses a big-bang singularity at $t=0$. This open set is
\begin{enumerate}[label=(\alph*)]
	\item matter dominated and asymptotically spatially homogeneous for $1 < \gamma < \frac{10}{9}$, $\frac{4}{3} < \gamma < \frac{10}{7}$, and $\frac{10}{7} < \gamma < \frac{3}{2}$.
	\item vacuum dominated and asymptotically spatially homogeneous for $1 < \gamma < \frac{3}{2}$.
	\item vacuum dominated and acceleration dominated for $1 < \gamma < \frac{3}{2}$.
\end{enumerate}
In each of these cases, the well-behaved models can be thought of as inhomogeneous generalizations of the associated spatially homogeneous equilibrium points. In all cases, except $\gamma = \frac{10}{9}$, the inhomogeneity adds an additional off-diagonal shear component. When $\gamma=\frac{10}{9}$, the off-diagonal shear component is already present and we are introducing inhomogeneity in the energy density and in this off-diagonal shear variable. Each one of the well-behaved cosmological models has a big-bang singularity at $t=0$, and has all of its physical and geometrical variables, energy density, acceleration, spatial curvature are bounded on any slice, $t =$constant $>0$.

\appendix 
\section{Asymptotic Analysis Near The Equilibrium Point $\mathrm{INF}^{\pm}$} 
In this appendix we show that cosmological models corresponding to trajectories which are asymptotic to the equilibrium points $\mathrm{INF}^{+}$ are badly-behaved in almost all cases. To be more precise, for the perfect fluid models, the acceleration of the fluid congruence is unbounded as $\chi \longrightarrow +\infty$. For the vacuum models, the components of the Weyl tensor diverge as $\chi \longrightarrow +\infty$ when $\gamma \neq \frac{4}{3}$. In the case of $\gamma = \frac{4}{3}$ the Weyl tensor components tend to zero and the asymptote is actually flat spacetime. The behaviour near $\mathrm{INF}^{-}$ is identical, and can be proven with the appropriate sign changes.\\

The equilibrium point $\mathrm{INF}^{+}$ has coordinates $Y_{1} = Y_{1e}$, $Y_{2} = 2(4-3\gamma) Y_{1e}$, $S =0$ with $Y_{1e} = \sqrt{ \frac{ 4(3-2\gamma)(\gamma-1)  }{(2-\gamma)^2}  }$. The eigenpairs of the linearization matrix of the DE (\ref{ODE-Y1Y2D2-1})-(\ref{ODE-Y1Y2D2-3}) at $\mathrm{INF}^{+}$ are
\begin{eqnarray}
\Bigl( 2 \alpha , \begin{bmatrix} \delta \\
		1  \\
		0 \\
	\end{bmatrix} \Bigl), 
\Bigl( 3 \alpha , \begin{bmatrix} 0 \\
	0 \\
	1 \\
\end{bmatrix} \Bigl),
\Bigl( \frac{(5 \gamma - 4)}{(\gamma - 1)} \alpha , \begin{bmatrix} 1 \\
	4-\gamma^2 - 2\gamma \\
	0 \\
\end{bmatrix} \Bigl),
\end{eqnarray}
where $\alpha = -4(\gamma-1)(2-\gamma)Y_{1e}$, $\delta = - \frac{(4-3\gamma)}{8(3-2\gamma)(\gamma-1)}$. The first two eigenpairs are the eigenpairs for the linearization of the DE restricted to the vacuum boundary. Close to $\mathrm{INF}^{+}$ the variables have the following behvaiour:
\begin{eqnarray}
	Y_{1} \approx Y_{1e}, \quad Y_{2} \approx 2(4 - 3\gamma) Y_{1e}, \quad S \approx e^{3 \alpha \chi}, \quad Y_{4} \approx e^{ (\frac{5 \gamma - 4}{\gamma - 1}) \alpha \chi}.  \label{Y1eAtEq}
\end{eqnarray}
Furthermore the variable $S \dot{U} \approx \frac{1}{2}(2-\gamma) Y_{1e}$ near $\mathrm{INF}^{+}$.\\
In order to obtain information about the physical variables we need to obtain the expansion scalar $\theta$. We have, from \cite[(A.43) and (A.53)]{hewitt2020},
\begin{eqnarray}
	\partial_{1} \theta = -r \theta = - 3 \dot{U} \theta.
\end{eqnarray}
And so, near $\mathrm{INF}^{+}$, we have
\begin{eqnarray}
	\frac{3}{8(\gamma-1)} \frac{1}{S} \frac{ \partial \theta}{\partial \chi} = - 3 \dot{U} \theta,
\end{eqnarray}
so that
\begin{eqnarray}
	\frac{\partial \theta}{\partial \chi} = -8(\gamma - 1) S \dot{U} \theta \approx \alpha \theta.
\end{eqnarray}
It follows that the spatial dependence of $\theta$ near $\mathrm{INF}^{+}$ is $\theta(\chi) \approx e^{\alpha \chi}$.\\
We can obtain this result another way for the perfect fluid models. Since, near $\mathrm{INF}^{+}$,
\begin{eqnarray}
	Y_{4} \approx e^{ \frac{(5 \gamma - 4)}{(\gamma - 1)} \alpha \chi   }, \mbox{and} \ S \approx e^{3\alpha \chi},
\end{eqnarray}
then
\begin{eqnarray}
	\Omega(\chi) =  Y_{4} S^{-1}\, \underline{\approx}\, e^{  \frac{(2-\gamma)}{(\gamma-1)} \alpha \chi  },
\end{eqnarray}
and so
\begin{eqnarray}
	\theta(\chi) = \Omega(\chi)^{ \frac{\gamma-1}{2-\gamma}  } \approx e^{\alpha \chi}. \label{thetaofX}
\end{eqnarray}
It now follows that the acceleration of the perfect fluid congruence, $\dot{u}_{1}$, is unbounded for all these models corresponding to trajectories approaching $\mathrm{INF}^{+}$:
\begin{eqnarray}
	\dot{u}_{1} = \theta \dot{U} \approx e^{\alpha \chi} e^{-3 \alpha  \chi} = e^{-2 \alpha \chi}, \label{udot1X}
\end{eqnarray}
which diverges as $\chi \longrightarrow \infty$.\\
We can determine the asymptotic behaviour of the physical variables for a cosmological model, corresponding to a trajectory approaching $\INF^{+}$, as follows. The asymptotic behaviour of the variables $Y_{1}$, $Y_{2}$, and $Y_{4}$ is given in (\ref{Y1eAtEq}), and (\ref{S-D}) and (\ref{CoordinateTransTUV-to-Y1Y2Y3Y4}) may be used to determine the asymptotic behaviour of the dimensionless variables. The asymptotic behaviour of the other physical variables is now determined using (\ref{thetaofX}) for $\theta$ and \cite[(26), (27)]{hewitt2020}. We have
\begin{eqnarray}
a_{1} \approx e^{-2\alpha \chi}, n_{13} \approx e^{-2\alpha \chi}, \sigma_{13} \approx e^{-2\alpha \chi}, \label{variable-2ax}
\end{eqnarray}
with
\begin{eqnarray}
	\lim_{\chi \to \infty } (n_{23} + a_{1}) =0, \label{Alimits1}\\
	\lim_{\chi \to \infty } (\dot{u}_{1} - 2 a_{1}) =0, \label{Alimits2}  \\
	\lim_{\chi \to \infty } (\sigma^2_{13} - 4a^2_{1} )  =0, \label{Alimits3}  \\
	\lim_{\chi \to \infty } \theta = 0. \label{Alimits4}
\end{eqnarray}
In order to make some statements about the vacuum models, we consider the non-zero Electric and Magnetic parts of the Weyl tensor \cite[p.19]{wainwright2005dynamicalSystems}, which are:
\begin{eqnarray}
\fl E_{11} = \frac{1}{2}(\gamma-1)(7\gamma - 10)\theta^2 + 2 a_{1} \dot{u}_{1} - \sigma^2_{13} + \frac{1}{3}(3 \gamma  -2) \mu, \label{E11-component}\\
\fl E_{22} = -(\gamma - 1)(3\gamma-4) \theta^2 - \dot{u}_{1}(a_{1}+ n_{23}) - \frac{1}{6}(3 \gamma - 2) \mu, \label{E22-component}\\
\fl E_{33} = - \frac{1}{2}(\gamma-1) (\gamma-2) \theta^2 - \dot{u}_{1}(a_{1}-n_{23}) + \sigma^{2}_{13} - \frac{1}{6}(3 \gamma -2) \mu, \label{E33-component}\\
\fl E_{13} = \theta \sigma_{13}  \frac{(4-3\gamma)}{2}, \label{E13-component}\\
\fl H_{12} = - (a_{1} + n_{23}) \sigma_{13}, \label{H12-component}\\
\fl H_{23} = \frac{\theta}{4}\Bigl( (5\gamma-6)(\dot{u}_{1} + a_{1}) + (3 \gamma-2) n_{23}  \Bigl).\label{H23-component}
\end{eqnarray}
If follows from (\ref{E13-component}), (\ref{variable-2ax}), and (\ref{thetaofX}) that if $\gamma \neq \frac{4}{3}$ then $E_{13} \approx e^{-\alpha \chi}$. Thus, $E_{13}$ and some of the invariants of the Weyl tensor diverge as $\chi \longrightarrow \infty$, and the models have a singularity.

The value of $\gamma$ being $\frac{4}{3}$ is important for the vacuum boundary for a number of reasons. First of all, the plane wave equilibrium points, LK, are actually the Bianchi type III form of flat spacetime ( see \cite{hsu1986self} ) for this value of the parameter. In addition, there is a pair of straight line solutions of the ODE (\ref{ODE-Y1Y2D2-1})-(\ref{ODE-Y1Y2D2-3}), given by $Y_{1} = 1$ and $Y_{1}=-1$. The corresponding cosmological model is just a slicing of flat spacetime using an accelerating timelike congruence. Thirdly,
when $\gamma = \frac{4}{3}$, $E_{13} = 0$ and, due to the limits indicated in (\ref{Alimits1})-(\ref{Alimits4}), the Weyl tensor components may tend to zero for vacuum models corresponding to trajectories tending to $\INF^{+}$ as $\chi \longrightarrow \infty$. It follows that for this value of $\gamma$, almost all the vacuum models are asymptotic to flat spacetime both as $\chi \longrightarrow \infty$ and as $\chi \longrightarrow -\infty$. The only exceptions being the RT model, and the vacuum models which are asymptotic to RT as either $\chi \longrightarrow \infty$ or as $\chi \longrightarrow -\infty$.
\ack
This work was partially supported by Natural Sciences And Engineering Council of Canada (NSERC) (B.C.) and partially by University of Waterloo's Department of Physics and Astronomy via Graduate Research Scholarships, Marie Curie Awards, and Science Graduate Awards (S.R.) and the C.E.M.C. (C.H.).

\newpage

\section*{References}
\bibliographystyle{iopart-num}
\bibliography{references2}

\end{document}